\documentclass[color,onefignum,onetabnum]{siamart220329}

\usepackage{graphicx,epstopdf} \usepackage[caption=false]{subfig} 
\usepackage{enumitem,cite,float}

\usepackage{lipsum}
\usepackage{amsfonts}
\usepackage{amsmath}
\usepackage{amssymb}
\usepackage{mathrsfs}
\usepackage{graphicx}
\usepackage[caption=false]{subfig}
\usepackage[export]{adjustbox}

\usepackage{epstopdf}
\usepackage{algorithmic}
\ifpdf
  \DeclareGraphicsExtensions{.eps,.pdf,.png,.jpg}
\else
  \DeclareGraphicsExtensions{.eps}
\fi

\renewcommand{\vec}[1]{\boldsymbol{#1}}
\newcommand{\dee}{\ensuremath{\textrm{ d}}}

\DeclareMathOperator{\supp}{supp}

\newsiamremark{remark}{Remark}
\newsiamremark{hypothesis}{Hypothesis}
\crefname{hypothesis}{Hypothesis}{Hypotheses}
\newsiamthm{claim}{Claim}
\newsiamthm{assumption}{Assumption}
\newsiamthm{example}{Example}

\headers{Modeling of electronic dynamics in twisted bilayer graphene}{T. Kong, D. Liu, M. Luskin and A. B. Watson}

\title{Modeling of electronic dynamics in twisted bilayer graphene
\thanks{Submitted to the editors \today. We are grateful to the two anonymous referees whose suggestions improved this manuscript significantly. \funding{TK's and ML's research was supported in part by Simons Targeted Grant Award No. 896630. DL's, AW's, and ML's research was supported in part by NSF DMREF Award No. 1922165. 
}}}

\author{Tianyu Kong\thanks{Department of Mathematics, University of Minnesota Twin Cities, MN (\email{kong0226@umn.edu})}
\and Diyi Liu\thanks{Department of Mathematics, University of Minnesota Twin Cities, MN (\email{liu00994@umn.edu})}
\and Mitchell Luskin\thanks{Department of Mathematics, University of Minnesota Twin Cities, MN (\email{luskin@umn.edu})}
\and Alexander B. Watson\thanks{Department of Mathematics, University of Minnesota Twin Cities, MN (\email{abwatson@umn.edu})\newline .}
}

\usepackage{amsopn}

\makeatletter
\newcommand*{\addFileDependency}[1]{
  \typeout{(#1)}
  \@addtofilelist{#1}
  \IfFileExists{#1}{}{\typeout{No file #1.}}
}
\makeatother

\newcommand*{\myexternaldocument}[1]{%
    \externaldocument{#1}%
    \addFileDependency{#1.tex}%
    \addFileDependency{#1.aux}%
}

\ifpdf
\hypersetup{
  pdftitle={Modeling of electronic dynamics in twisted bilayer graphene},
  pdfauthor={T. Kong, D. Liu, M. Luskin and A. B. Watson}
}
\fi

\myexternaldocument{supplement}

\begin{document}

\maketitle

\begin{abstract}
    We consider the problem of numerically computing the quantum dynamics of an electron in twisted bilayer graphene. The challenge is that atomic-scale models of the dynamics are aperiodic for generic twist angles because of the incommensurability of the layers. The Bistritzer-MacDonald PDE model, which is periodic with respect to the bilayer's moir\'e pattern, has recently been shown to rigorously describe these dynamics in a parameter regime.     In this work, we first prove that the dynamics of the tight-binding model of incommensurate twisted bilayer graphene can be approximated by computations on finite domains. The main ingredient of this proof is a speed of propagation estimate proved using Combes-Thomas estimates. We then provide extensive numerical computations which clarify the range of validity of the Bistritzer-MacDonald model.
\end{abstract}




\section{Motivation and summary} \label{sec:mot_and_sum}

In recent years, twisted bilayer graphene and other stackings of 2D materials have emerged as important experimental platforms for realizing quantum many-body phases such as superconductivity \cite{Cao_Fatemi_Demir_Fang_Tomarken_Luo_Sanchez-Yamagishi_Watanabe_Taniguchi_Kaxiras_etal_2018,Cao_Fatemi_Fang_Watanabe_Taniguchi_Kaxiras_Jarillo-Herrero_2018}. These developments were made possible by Bistritzer and MacDonald's observation that the single-particle electronic properties of stackings with relatively small lattice mismatches (for example, layers of the same 2D material with a small twist angle) can often be captured by effective continuum models which are periodic over the stacking's moir\'e pattern \cite{Bistritzer_MacDonald_2011}. This observation meant that, despite 2D materials stackings often being aperiodic at the atomic scale (for example, layers of the same 2D material with an irrational twist angle), their properties could be studied using ordinary Bloch band theory. 

This theoretical simplification allowed, for example, for the identification of bilayer graphene's ``magic'' twist angle, $\theta \approx 1.05^\circ$. Near to this angle, the single-particle continuum model Bloch bands (dispersion relation) at the Fermi level become very flat \cite{Bistritzer_MacDonald_2011}. Based on this observation, Bistritzer and MacDonald predicted that electrons in twisted bilayer graphene interact relatively strongly at this twist angle, resulting in a rich quantum many-body phase diagram. This prediction was dramatically verified in the experiments \cite{Cao_Fatemi_Demir_Fang_Tomarken_Luo_Sanchez-Yamagishi_Watanabe_Taniguchi_Kaxiras_etal_2018,Cao_Fatemi_Fang_Watanabe_Taniguchi_Kaxiras_Jarillo-Herrero_2018}.

The importance of the theoretical simplification provided by effective continuum models motivates the question of their range of validity. This question was recently considered in detail by three of the authors of this work for the special case of the Bistritzer-MacDonald continuum model of twisted bilayer graphene \cite{Watson_Kong_MacDonald_Luskin_2023}. They considered an atomic-scale tight-binding Schr\"odinger model governing the dynamics of the wave-function of a single electron in twisted bilayer graphene, in the absence of mechanical relaxation, with wave-packet initial data spectrally concentrated at the monolayer Dirac points. Then, they estimated the difference at time $t > 0$, in the natural $\ell^2$ norm, between the wave-packet time-evolved according to the tight-binding model $\psi_{\text{TB}}(t)$, and the same wave-packet time-evolved according to the Bistritzer-MacDonald model $\psi_{\text{BM}}(t)$. 

The main result of \cite{Watson_Kong_MacDonald_Luskin_2023} can be summarized simply as
\begin{equation} \label{eq:main_estimate}
    \| \psi_{\text{TB}}(t) - \psi_{\text{BM}}(t) \|_{\ell^2} \leq \rho\left(\frac{\theta}{\epsilon}\right) \times \left( \epsilon^2 + \epsilon \theta + \epsilon \mathfrak{h}^{1-c_1} + \mathfrak{h}^{2-c_2} \right) \times t.
\end{equation}
Here, $\epsilon, \theta$, and $\mathfrak h$ are dimensionless parameters, and $\rho(\xi)$ denotes a positive continuous function which tends to $\infty$ as $\xi \rightarrow \infty$, and converges to a constant as $\xi \rightarrow 0$. In particular, $\rho\left(\frac{\theta}{\epsilon}\right)$ can be uniformly bounded as long as $\frac{\theta}{\epsilon}$ remains bounded. The parameter $\epsilon$ denotes the spectral width of the wave-packet in momentum space normalized by the monolayer lattice constant, $\theta$ the twist angle in radians, and $\mathfrak h$ the ratio of the largest interlayer hopping energy in momentum space to the largest intralayer hopping energy in real space. For realistic choices of the interlayer hopping function, the constants $c_1, c_2 > 0$ can be taken arbitarily small. It follows immediately from \cref{eq:main_estimate} that
\begin{equation} \label{eq:parameter_regime}
    \theta \lesssim \epsilon \text{ and } \mathfrak h \sim \epsilon \implies \| \psi_{\text{TB}}(t) - \psi_{\text{BM}}(t) \|_{\ell^2} \leq C \epsilon^{2 - c} t,
\end{equation}
where $C, c > 0$ are constants independent of $\epsilon$ and $t$, and $c$ can be taken arbitarily small. Hence, in parameter regime \cref{eq:parameter_regime}, the Bistritzer-MacDonald model captures the dynamics of the tight-binding model up to times $\sim \epsilon^{- (2 - \delta)}$ for any $\delta > 0$. 

It is natural to ask whether this regime is realized in experiments. The magic angle corresponds to $\theta \approx 0.017$ radians, while the value of $\mathfrak{h}$ is estimated as $\approx 0.042$. It follows that the regime \cref{eq:parameter_regime} is indeed realized, for a non-trivial range of $\epsilon$, at the magic angle. It should be emphasized that rigorous justification of any moir\'e-scale model for the many-body electronic properties of twisted bilayer graphene is a challenging open problem, although a formal plausibility argument for such reductions is provided in \cite{Watson_Kong_MacDonald_Luskin_2023}. The arguments provided can partially justify the single-particle Hamiltonian term in an interacting Bistritzer-MacDonald model of TBG in \cite{faulstich2022interacting}.

It is currently unclear whether estimate \cref{eq:main_estimate}, proved in \cite{Watson_Kong_MacDonald_Luskin_2023}, is sharp. In particular, the following questions regarding the convergence of the tight-binding dynamics to those of the Bistritzer-MacDonald model were not answered by \cite{Watson_Kong_MacDonald_Luskin_2023}:
\begin{enumerate}[label={(\arabic*)}]
    \item How well does the Bistritzer-MacDonald approximation perform as a practical matter, both in the regime \cref{eq:parameter_regime} and otherwise? The point here is that the error in \cref{eq:parameter_regime} could be large in practice, even in the regime \cref{eq:parameter_regime}, if the constant $C$ is large.
    \item Suppose we start in the regime \cref{eq:parameter_regime}, and then increase the parameters $\epsilon, \theta$, and $\mathfrak h$ individually. Does \cref{eq:main_estimate} capture the correct dependence of the error on each parameter?
    \item Is \cref{eq:parameter_regime} the only regime where the Bistritzer-MacDonald model captures the dynamics of the tight-binding model? In other words, outside of the regime \cref{eq:parameter_regime}, is the error always large?
\end{enumerate}
The focus of the present work is to begin to address questions (1)-(3) by accurate numerical computation of time-evolved wave-packet solutions of the tight-binding model of twisted bilayer graphene.

Our numerical experiments provide the following (roughly stated) answers to these questions:
\begin{enumerate}[label={(\arabic*)}]
    \item In the regime \cref{eq:parameter_regime}, the Bistritzer-MacDonald approximation does indeed capture critical features of the tight-binding dynamics. For example, we find that the band structure of the Bistritzer-MacDonald model does predict the group velocity of spectrally concentrated wave-packet solutions of the tight-binding model; see Figures \ref{fig:BM_band_structure} and \ref{fig:TB_BM_error}. In particular, at the magic angle, the group velocity is essentially zero (Figure \ref{fig:BM_flat_band}). A plot of the approximation error as a function of $\epsilon$, with $\theta$ and $\mathfrak{h}$ scaled according to \eqref{eq:parameter_regime}, is provided in Figure \ref{fig:BM_epsilon_error}.
    \item The estimate \cref{eq:main_estimate} generally does capture the correct scaling of the error as a function of each parameter; see Figure \ref{fig:BM_param_error}. The exception is the dependence of the error on $\theta$, where we find that, for large $\theta$, the error is much smaller than predicted by the estimate. Instead of growing, the error remains small with essentially constant size as $\theta$ is increased.
    \item We do not exhaustively investigate all possible parameter regimes, but Figure \ref{fig:BM_param_error} suggests that the Bistritzer-MacDonald approximation remains accurate even for relatively large $\theta$, as long as $\epsilon$ and $\mathfrak{h}$ are small.
\end{enumerate}
We discuss the details of our numerical experiments in Section \ref{sec:num_trunc}.

An alternative justification of the Bistritzer-MacDonald model has been provided in \cite{Cances_Garrigue_Gontier_2023} (see also \cite{cances2021secondorder}). The starting point of their work is a continuum Kohn-Sham DFT description of the twisted bilayer. They show that it is possible to pass to a moir\'e-periodic continuum model, all of whose parameters can be numerically computed via DFT applied to untwisted layers, under fairly general assumptions. That moir\'e-periodic continuum model has additional terms compared with the model originally proposed by Bistritzer and MacDonald in \cite{Bistritzer_MacDonald_2011}, but numerical computations of these terms at realistic model parameters (in particular, at realistic values of the twist angle and interlayer distance) find that these terms are small \cite{Cances_Garrigue_Gontier_2023}.

The accurate numerical computation of time evolved solutions of the tight-binding model of twisted bilayer graphene is made challenging by the fact that the model is infinite dimensional (the Hilbert space is isomorphic to $\ell^2(\mathbb{Z}^2)$) and aperiodic at generic twist angles. A standard approach to obtaining a finite dimensional model for computation is to approximate the twist angle by a rational angle, so that the system can be treated as periodic. Such approaches are known as supercell approximations \cite{Lin_Lu_2019}. 

An alternative approach is to leave the twist angle fixed, but truncate the computational domain (i.e., impose a Dirichlet boundary condition), far from the support of the initial data. We follow the second approach in the present work, because with this approach we can rigorously estimate the difference between the dynamics of the truncated model and those of the untruncated model at any twist angle of interest. Similar ideas have been used for numerical computation of dynamics with error control \cite{Colbrook_Horning_Thicke_Watson_2021, Colbrook_semigroup}, although in those works the truncation distance is chosen adaptively, while we give an {\em a priori} estimate. 

The main idea of the estimate is a Lieb-Robinson bound\cite{LiebRobinson1972}, i.e., a bound on the speed of propagation for solutions of the tight-binding Schr\"odinger equation (up to error which is exponentially small in the distance). To keep our work self-consistent, we give a straightforward proof of the Lieb-Robinson bound we require using Combes-Thomas estimates \cite{1973CombesThomas}. The study of Lieb-Robinson bounds for quantum many-body systems remains an active area; see \cite{Chen_Lucas_Yin_2023, Hastings_2012} and references therein.

Note that computing spectrally-concentrated wave-packet solutions of the truncated system is still difficult, because such solutions spread over the moir\'e cell (length $\propto \theta^{-1}$), necessitating large domain truncations. A layer-splitting numerical scheme was recently proposed to compute dynamics of incommensurate heterostructures in \cite{Wang_Chen_Zhou_Zhou_2021}. This work built on related work applying plane wave decomposition to compute other properties of such heterostructures in \cite{ZHOU201999,Wang2023,Massatt_Carr_Luskin_2021,momentumspace17}. We aim to combine these ideas with those of the present work to obtain an efficient numerical method with rigorous error estimates in future work.

\subsection{Description of results}

We now briefly describe the results of this work. We first describe the analytical results, which prove convergence of our tight-binding numerical computations on finite computational domains to solutions of the model without truncation. We will then describe our computational results.

Our first analytical result, \cref{thm:combes-thomas}, is a Combes-Thomas estimate on decay of the matrix elements of the resolvent of the tight-binding Hamiltonian. We then use this estimate to prove a Lieb-Robinson bound on the speed of propagation in \cref{prop:speed_of_propagation}. This bound allows us to prove convergence, at fixed time $t$, of solutions of the truncated tight-binding model to those of the untruncated model, up to exponentially small error in the truncation length, in \cref{thm:trunc_main}. We confirm the exponentially fast convergence of our truncated domain computations as the size of the truncation is increased computationally in \cref{fig:trunc_R}.

We now describe the results of our numerical comparisons between tight-binding dynamics and those generated by the Bistritzer-MacDonald model. In \cref{fig:TB_BM_error}, we compare these dynamics directly, for initial conditions spectrally concentrated in higher (not flat) Bloch bands of the Bistritzer-MacDonald model, so that the wave-packet has a clear non-zero group velocity (we show the band structure of the Bistritzer-MacDonald model in \cref{fig:BM_band_structure}). The results confirm that the continuum model accurately captures the most obvious features of the tight-binding dynamics, although clear errors can be seen even for relatively small times. In \cref{fig:BM_flat_band}, we repeat the same experiments but for wave-packets concentrated in the flat bands of the Bistritzer-MacDonald model. We find that the group velocity of wave-packets is negligible, as is to be expected, but also that the Bistritzer-MacDonald model misses interesting features of the tight-binding solution. Specifically, the Bistritzer-MacDonald model appears to miss a twist-angle-dependent chirality of the solution (see Figure 
\ref{fig:BM_flat_band} and caption).  

In \cref{fig:BM_epsilon_error}, we confirm that, for sufficiently small values of the parameters scaled according to \cref{eq:parameter_regime}, and sufficiently small $t$, the form of the error is indeed $C \epsilon^2 t$. In \cref{fig:BM_param_error}, we start in the regime \cref{eq:parameter_regime}, and then vary each of the parameters $\epsilon, \mathfrak h, \theta$ individually. In the case of $\mathfrak h$, we confirm that the error grows linearly, and in the case of $\epsilon$, we confirm that the error grows between linearly and quadratically. Our most interesting result is in the case of $\theta$, where we observe that the error is essentially constant as $\theta$ as increased, suggesting a wider range of applicability of the Bistritzer-MacDonald model than could be expected from the results of \cite{Watson_Kong_MacDonald_Luskin_2023}. We aim to provide an analytical explanation of this phenomenon in future work. 



\subsection{Structure of this paper}

The structure of the remainder of our paper is as follows. We first introduce the lattice structure of  monolayer and twisted bilayer graphene in \cref{sec:tbg}.  We then define the tight-binding Hamiltonian and its finite dimensional approximation through domain truncation in \cref{sec:tight_binding}, and present our estimate on the truncation error (\cref{thm:combes-thomas}, \cref{prop:speed_of_propagation}, and \cref{thm:trunc_main}). In \cref{sec:BM}, we review the continuum approximation of TBG, the Bistritzer-MacDonald model, and recall the main result of \cite{Watson_Kong_MacDonald_Luskin_2023} on the parameter regime in which the approximation error can be estimated (\cref{thm:BM}).

We present several numerical results to validate the truncation error of the tight-binding model in \cref{sec:num_trunc}. We then present our results directly comparing the dynamics of the Bistrizer-MacDonald model and of the tight-binding model across various initial conditions in \cref{sec:num_BM}. Finally in \cref{sec:num_BM_beyond} we numerically compute the sensitivity of the error as a function of the model parameters. The proofs and the technical details for this paper are presented in the Appendices.

\subsection{Code availability}

We have made the code used to generate our numerical results available at \verb|github.com/timkong98/dynamics_tbg|. 

\section{Quantum dynamics of twisted bilayer graphene}


In this section, we recall the tight-binding model of twisted bilayer graphene studied in \cite{Watson_Kong_MacDonald_Luskin_2023}. 

\subsection{Twisted bilayer graphene}\label{sec:tbg}

Graphene is a single sheet of carbon atoms arranged in a honeycomb structure. 
Each unit cell contains two atoms, and the unit cells form a Bravais lattice with vectors
\begin{equation}
    \vec a_1 := \frac{a}{2}(1,\sqrt{3})^\top, \quad \vec a_2 := \frac{a}{2}(-1,\sqrt{3})^\top \quad A := (\vec a_1, \vec a_2),
\end{equation}
where $a$ is the lattice constant. The physical value of the graphene lattice constant is approximately $a \approx 2.5 \text{ \AA}$. The graphene Bravais lattice $\mathcal{R}$ and a unit cell $\Gamma$ can be defined as
\begin{equation}
    \mathcal{R} := \{ \vec R = A\vec n: \vec n\in \mathbb{Z}^2 \}, \quad 
    \Gamma=\{ A \alpha : \alpha \in [0,1)^{2} \}.
\end{equation}
Within a unit cell indexed by $\vec R$, there are two atoms at physical location $\vec R + \vec \tau^A$ and $\vec R + \vec\tau^B$, which we define as
\begin{equation} \label{eq:relative_shift}
\vec\tau^A := (0,0)^\top, \quad \vec\tau^B : = \left(0, \delta \right)^\top, \quad \delta:= \frac{a}{\sqrt{3}}.
\end{equation}
These atoms are in sub-lattices $A$ and $B$ respectively, and the relative shift between two sub-lattices $\delta$ is the minimum distance between two atoms in the same layer.

The reciprocal lattice vectors are defined through the relation $\vec a_i \cdot \vec b_j = 2\pi\delta_{ij}$ for $\delta_{ij}$ the Kronecker delta and $i, j\in\{1,2\}$. Explicitly they are
\begin{equation}
    \vec{b}_1 := \frac{4 \pi}{3 \delta} \left( \frac{\sqrt{3}}{2},\frac{1}{2} \right)^\top, 
    \quad \vec{b}_2 := \frac{4 \pi}{3 \delta} \left( - \frac{ \sqrt{3} }{2},\frac{1}{2} \right)^\top, 
    \quad B := ( \vec{b}_1, \vec{b}_2 ).
\end{equation}
Similarly we define the reciprocal lattice $\Lambda^*$ and a fundamental cell $\Gamma^*$ by
\begin{equation}
    \Lambda^* := \left\{ \vec{G} = B \vec{n} : \vec{n} \in \mathbb{Z}^2 \right\}, \quad \Gamma^* := \left\{ B \beta : \beta \in \left[0,1\right)^2 \right\}.
\end{equation}
The Dirac points of graphene are 
\begin{equation}
    \vec{K} := \frac{ 4 \pi }{ 3 a } ( 1, 0 )^\top, \quad \vec{K}' := - \vec{K}.
\end{equation}

Twisted bilayer graphene (TBG) consists of two monolayer graphene in parallel planes with a relative twist angle, and separated by an interlayer distance $L$. In particular for two layers of rigid graphene, each layer can be described by a rotated Bravais lattice. Let $R(\eta)$ be the matrix that describes a counter-clockwise rotation by $\eta$ around the origin,
\begin{equation}
        R(\eta) := \begin{pmatrix} 
        \cos \eta & - \sin \eta \\ 
        \sin \eta &  \cos \eta 
        \end{pmatrix}.
\end{equation}
Then for any twist angle $\theta > 0$, we can define the lattice vectors of TBG by
\begin{equation}
    \vec a_{1,i} := R\left(-\frac{\theta}{2}\right)\vec a_i, \quad
   \vec a_{2,i} := R\left(\frac{\theta}{2}\right)\vec a_i, \quad
    A_j := (\vec a_{j,1}, \vec a_{j,2}), \quad i \in \{1,2\}, \, j\in\{1,2\}.
\end{equation}
Here $j$ describes the layer, and $i$ describes the lattice vector in each layer. The relative shift between sublattices are
\begin{equation}
\vec\tau^\sigma_1 := R\left(-\frac{\theta}{2}\right) \vec\tau^\sigma, \quad \vec\tau^\sigma_2 := R\left(\frac{\theta}{2}\right) \vec\tau^\sigma, \quad 
\sigma \in \{A, B\},
\end{equation}
and the lattices are  
\begin{equation}
    \mathcal{R}_j := \left\{\vec R_j = A_j \vec n,\, \vec n\in \mathbb{Z}^2 \right\}, \quad j\in\{1,2\}.
\end{equation} 
Similarly, the reciprocal lattice vectors of TBGs are
\begin{equation}
    \vec{b}_{1,i} := R\left(- \frac{\theta}{2}\right) \vec{b}_i, \quad \vec{b}_{2,i} := R\left(\frac{\theta}{2}\right) \vec{b}_i, \quad B_j := (\vec{b}_{j,1}, \vec{b}_{j,2}) \quad i \in \{1,2\}, \, j\in\{1,2\}.
\end{equation}
The $\vec K$ and $\vec K'$ points of each layer are
\begin{equation}
    \vec{K}_1 := R\left(-\frac{\theta}{2}\right) \vec{K}, \quad \vec{K}_2 := R\left(\frac{\theta}{2}\right) \vec{K}, \quad \vec{K}_i' := - \vec{K}_i, \quad i \in \{1,2\}.
\end{equation}



For each layer, the lattice $\mathcal{R}_j$ is a rotated monolayer Bravais lattice, thus also periodic. For general twist angle $\theta$, the periodicity is broken in the bilayer system $\mathcal{R}_1 \cup \mathcal{R}_2$.  
Even though TBG is not exactly periodic, there is an approximate periodicity known as the moir\'e pattern (see \cref{fig:moire_potential}). The moir\'e reciprocal lattice vectors are given by  the difference of reciprocal lattice vectors between layers \cite{Cazeaux_Luskin_Massatt_2020,relaxfosdick22}
\begin{equation}
    \vec{b}_{m,1} := \vec{b}_{1,1} - \vec{b}_{2,1}, \quad \vec{b}_{m,2} := \vec{b}_{1,2} - \vec{b}_{2,2}.
\end{equation}
These vectors can be computed explicitly. Let $|\Delta \vec K| := | \vec K_1  - \vec K_2|= 2 | \vec{K} | \sin\left( \frac{\theta}{2} \right)$ be the distance between the Dirac points of the layers. Then, we have
\begin{equation} 
    \vec{b}_{m,1} = \sqrt{3} | \Delta \vec{K} | \left( \frac{1}{2}, -\frac{\sqrt{3}}{2} \right)^\top, \quad \vec{b}_{m,2} = \sqrt{3} | \Delta \vec{K} | \left( \frac{1}{2} , \frac{\sqrt{3}}{2} \right)^\top.
\end{equation}
The moir\'e lattice vectors are defined through the relation  $\vec a_{m,i} \cdot \vec b_{m.j} = 2\pi\delta_{ij}$ for $i, j\in\{1,2\}$,
\begin{equation}
    \vec{a}_{m,1} := \frac{4 \pi}{3 | \Delta \vec{K} |} \left( \frac{\sqrt{3}}{2}, -\frac{1}{2} \right)^\top, \quad \vec{a}_{m,2} := \frac{4 \pi}{3 | \Delta \vec{K} |} \left( \frac{\sqrt{3}}{2} , \frac{1}{2} \right)^\top.
\end{equation}
Defining $A_m := (\vec a_{m,1}, \vec a_{m,2})^\top$, and  $B_m := (\vec b_{m,1}, \vec b_{m,2})^\top$, the moir\'e lattice, unit cell, reciprocal lattice, and reciprocal unit cell are analogous to their monolayer counterparts
\begin{equation}
\begin{gathered}
    \Lambda_m := \{ \vec{R}_m = A_m \vec{m} : \vec{m} \in \mathbb{Z}^2 \}, \quad \Gamma_m := \left\{ A_m \alpha : \alpha \in \left[ 0 , 1 \right)^2 \right\},\\
     \Lambda_m^* := \{ \vec{G}_m = B_m \vec{n} : \vec{n} \in \mathbb{Z}^2 \}, \quad \Gamma_m^* := \left\{ B_m \beta : \beta \in \left[ 0 , 1 \right)^2 \right\}.
\end{gathered}
\end{equation}

When the twist angle $\theta$ is small, the length of the moir\'e lattice vectors $| \vec a_{m,i}|$ is proportional to $a\theta^{-1}$, significantly longer than those of monolayer graphene. To model the physical phenomenon on a moir\'e scale for small $\theta$, we will need to consider including at least thousands of atoms in the model. This gives a rough estimate of the truncation radius when studying electronic dynamics in TBG. 



\subsection{The tight-binding Hamiltonian}\label{sec:tight_binding}

In this section we introduce a natural tight-binding Hamiltonian \cite{Kaxiras_Joannopoulos_2019, Lin_Lu_2019, Ashcroft_Mermin_1976} for an electron in TBG. This model is natural in the sense that it trades off some accuracy for considerable conceptual and computational simplification compared with fundamental continuum PDE Schr\"odinger equation models. Tight-binding models arise as Galerkin approximations to continuum PDE models; when parametrized using careful DFT computations, tight-binding models of twisted heterostructures derived using Wannier basis orbitals have comparable accuracy to large-scale DFT computations at fixed commensurate twist angles \cite{PhysRevB.93.235153,Carr_Fang_Kaxiras_2020}. Rigorous derivations of such models were provided in \cite{1988HelfferSjostrand,2017FeffermanLeeThorpWeinstein}.


First we precisely define the Hamiltonian and wave functions in these systems.
Let $\mathcal{A}_j = \{A, B\}$ denote the set of indices of orbitals associated with each unit cell in layer $j$. Then the full degree of freedom space of TBG can be described using an index set
\begin{equation}
    \Omega := (\mathcal{R}_1\times \mathcal{A}_1) \cup (\mathcal{R}_2\times \mathcal{A}_2).
\end{equation}
For the atom indexed by $\vec R_i\sigma \in \Omega$, where $i\in \{1,2\}$ denotes the layer, and $\sigma \in \{A, B\}$ denotes the sublattice, the physical location is $\vec R_i + \vec\tau_i^\sigma \in \mathbb{R}^2$. Note that assuming a constant interlayer distance allows us to model TBG by a 2D model, while modeling the effect of non-zero interlayer distance through the interlayer hopping function.


We model the wave function of an electron in TBG as an element of the Hilbert space
\begin{equation}
\begin{gathered}
    \mathcal{H}  := \ell^2(\Omega) = \left\{ (\psi_{\vec R_i\sigma})_{ \vec R_i\sigma \in \Omega } : 
    \| \psi \|_\mathcal{H} < \infty\right\}, \\
    \| \psi \|_\mathcal{H} =  \left(\sum_{i \in\{1,2\}} \sum_{\sigma \in\{A,B\}} \sum_{\vec R_i\in \mathcal{R}_i}| \psi_{\vec R_i\sigma}|^2\right)^\frac{1}{2}.
    \end{gathered}
    \end{equation}
   For ease of notation, we write the three summations as a single summation over elements of the index set. The square of the modulus of $\psi_{\vec R_i\sigma}$ represents the electron density on the orbital of sublattice $\sigma \in\{A, B\}$ in the $\vec R$th cell of layer $i \in\{1,2\}$. 

We define the TBG tight-binding Hamiltonian $H: \mathcal{H} \to \mathcal{H}$ to be a linear self-adjoint operator that acts on the wave functions as
\begin{equation} \label{eq:H_tb}
    (H\psi)_{\vec R_i\sigma} = \sum_{\vec R_j\sigma' \in \Omega} H_{\vec R_i\sigma, \vec R'_j\sigma'}\psi_{\vec R_j'\sigma'}, 
\; H_{\vec R_i\sigma, \vec R'_j\sigma'} = \overline{H_{\vec R'_j\sigma', \vec R_i\sigma}}.
\end{equation}
 We make the following exponential decay assumption to ensure the Hamiltonian is localized. Exponential decay is natural since the tight-binding model is defined through exponentially localized Wannier functions.
\begin{assumption}[Exponential decay hopping]
\label{as:decay}
There exist constants $h_0, \alpha_0 > 0$ such that 
\begin{equation}
|H_{\vec R_i\sigma, \vec R'_j\sigma'}| \leq h_0 e^{-\alpha_0 \left|\vec R_i + \vec \tau_i^\sigma - \vec R'_j - \vec\tau_j^{\sigma'}\right|}.
\end{equation}
\end{assumption}

\begin{lemma}
\label{lem:bounded}
Under \cref{as:decay}, $H$ is a bounded self-adjoint operator. Its operator norm is bounded by
\begin{equation}
    \|H\| \leq \frac{8\pi h_0 e^{\delta\alpha_0}  }{|\Gamma| \alpha_0^2},
\end{equation}
where $\delta = a/\sqrt{3}$, and $|\Gamma| = \sqrt{3}a^2/2$ is the area of a unit cell.
As a consequence, the spectrum $\sigma(H)$ is contained in $ [-\|H\|, \|H\|] \subset \mathbb R$.
\end{lemma}
\begin{proof}
See \cref{prf:h_bound}.
\end{proof}

\begin{example}
\label{ex:nearest_neighbor}
    The specific tight-binding model of TBG studied in \cite{Watson_Kong_MacDonald_Luskin_2023} is as follows.  For orbitals on the same layer, the entries are non-zero except when they are the nearest neighbors in the lattice. For some $t_0 > 0$, we have
\begin{equation}
    H_{\vec R_i\sigma, \vec R'_i\sigma'} = 
    \begin{cases}
        -t_0, \; \text{if} \quad \left|\vec R_i + \vec\tau_i^\sigma - \vec R'_i - \vec\tau_i^{\sigma'}\right| = \delta, \\
        0, \quad \text{otherwise.}
    \end{cases}
\end{equation}
    The restriction to nearest-neighbor hopping makes the analysis particularly simple. Since the magnitude of next-nearest-neighbor hopping terms are thought to be $\approx .1 t_0$ \cite{2009Castro-NetoGuineaPeresNovoselovGeim}, we do not expect that neglecting longer-range hops changes any essential model phenomena.

For orbitals on different layers, we can define the entries using an interlayer hopping function that also encodes the interlayer distance $L$. For some $h_0, \alpha>0$, we have
\begin{equation}
     H_{\vec R_i\sigma, \vec R'_j\sigma'} = h(\vec R_i + \vec\tau_i^\sigma - \vec R'_j - \vec\tau_j^{\sigma'}; L), \quad h(\vec r; L) = h_0 e^{-\alpha_0 \sqrt{|\vec r|^2 + L^2}}.
\end{equation}
    A simple calculation shows this Hamiltonian satisfies \cref{as:decay}. Note that the Fourier transform of the interlayer hopping function is \cite{Bateman1954a} 
\begin{equation} \label{eq:h_hat}
\hat h(\vec \xi; L) = 2\pi h_0 \frac{\alpha_0 e^{-L\sqrt{|\vec\xi|^2+\alpha_0^2}}\left(1+L\sqrt{|\vec\xi|^2+\alpha_0^2}\right)}{\left(|\vec\xi|^2+\alpha_0^2\right)^{3/2}}.
\end{equation}
\end{example}

We now identify the dynamics of wave functions in TBG as the solutions to the initial value problem of the time-dependent Schr\"odinger equation
\begin{equation}
\label{eq:dynamic_hbar}
        i\hbar\partial_t \psi = H\psi, \quad \psi(0) = \psi_0 \in \mathcal{H}. 
\end{equation}
We set $\hbar = 1$ for convenience. The solution at time $t$ can be expressed using holomorphic functional calculus. Since $H$ has bounded spectrum, we can find a suitable Jordan curve $\gamma=\partial D$, where $D$ is a simply connected domain such that $\sigma(H)\cap \gamma=\emptyset$ and $\sigma(H)\subset D$. Then Cauchy's integral formula gives the following well-known identity in terms of the Bochner integral \cite{yosida2012functional}
\begin{equation} \label{eq:prop}
    \psi(t) = e^{-iHt}\psi_0 = \frac{1}{2\pi i}\int_\gamma e^{-itz}(z-H)^{-1}\psi_0 \dee z.
\end{equation}
Equation \eqref{eq:prop} connects the propagator to the resolvent, to which Combes-Thomas estimates can be applied; see Theorem \ref{thm:combes-thomas} and \cite{1973CombesThomas}. This idea was also used in the works \cite{Colbrook_Horning_Thicke_Watson_2021,Colbrook_semigroup} on computing dynamics with error control. 
\begin{remark}
    In fact, this model is not limited to twisted bilayer graphene with periodic monolayers. For example, when stacking one layer of graphene on another, it is experimentally observed that carbon atoms undergo a small displacement to minimize the total energy, a phenomenon called mechanical relaxation 
    \cite{Cazeaux_Luskin_Massatt_2020, Massatt_Carr_Luskin_2021, Yoo_Engelke_Carr_Fang_Zhang_Cazeaux_Sung_Hovden_Tsen_Taniguchi_etal_2019}. The relaxation can be added to the system through a displacement function $u : \Omega \to \mathbb R^2$, and the physical location of an atom indexed by $\vec R_i\sigma$ in a relaxed TBG system is $\vec R_i + \vec \tau_i^\sigma + u(\vec R_i\sigma)$. 
\end{remark}
\begin{remark}
    Our approach can also be applied to general aperiodic systems in $n$ dimensions. For these systems, we only assume that there is no accumulation point for the physical locations of orbitals. As long as the Hamiltonian is localized as in \cref{as:decay}, we are able to prove a similar estimate as in this paper. 
\end{remark}

The lattices of TBG are infinite, and it's impossible to numerically compute the dynamics of the infinite system.
We transform the infinite system into a finite system through domain truncation.  This method was used to calculate other observables in TBG, such as the local density of states \cite{Massatt_Luskin_Ortner_2017}. 
First, let $\Omega_{R}$ be the finite subset of atomic orbital indices inside a ball $B_R$,
\begin{equation}\begin{split}
    \Omega_{R} := \{ \vec R_i\sigma \in \Omega : |\vec R_i + \vec \tau_i^\sigma| \leq R\}.
\end{split}\end{equation}
We can define the finite dimensional injection map $P_{R}: \mathcal H \to \ell^2(\Omega_R)$ along with its adjoint $P_{R}^*:  \ell^2(\Omega_R) \to \mathcal H$
\begin{equation}
    P_R \psi = \left(\psi_{\vec R_i\sigma}\right)_{\vec R_i\sigma \in \Omega_R}, \quad 
    \left(P_R^*\Psi\right)_{\vec R_i \sigma} = \begin{cases}
        \Psi_{\vec R_i\sigma}, \; \vec R_i\sigma \in \Omega_R, \\
        0, \;\text{otherwise}.
    \end{cases}
\end{equation}
The finite dimensional restriction on the Hamilonian is
\begin{equation}
\label{eq:H_trunc}
H_R := P_R H P_R^*,
\end{equation}
which is a $|\Omega_R| \times |\Omega_R|$ Hermitian matrix. The truncated $H_R$ ignores all interactions with sites outside a ball of radius $R$.
The truncation is only valid when the wave function is spatially concentrated inside the ball $B_R$. To make sure this is true over a period of time, we assume the initial condition is concentrated on a smaller ball $B_r$ with radius $r < R$.

For any set $A\subset \mathbb{R}^2$, we define $\mathcal{X}_A$ as the characteristic function 
\begin{equation}
 \left(\mathcal{X}_A \right)_{\vec R_i\sigma} = 
 \begin{cases}
     1, \; \text{if } \vec R_i + \tau_i^\sigma \in A, \\
     0, \;\text{ otherwise.}
 \end{cases}
\end{equation}
We thus have
\begin{equation}
 \left(\mathcal{X}_A \psi\right)_{\vec R_i\sigma} = 
 \begin{cases}
     \psi_{\vec R_i\sigma}, \; \text{if } \vec R_i + \tau_i^\sigma \in A, \\
     0, \;\text{ otherwise,}
 \end{cases}
\end{equation}
with the properties
\begin{equation}
\mathcal{X}_A + \mathcal{X}_{A^\complement} = I, \quad \|\psi\|_\mathcal{H} = \left\| \mathcal{X}_A\psi \right\|_\mathcal{H} + \left\| \mathcal{X}_{A^\complement} \psi \right\|_\mathcal{H}.
\end{equation}

\begin{assumption}[Decay of the initial condition] 
\label{as:intial_condition}
The initial condition $\psi_0$ of the initial value problem \cref{eq:dynamic_hbar} satisfies
\begin{equation}
        \left\| \mathcal{X}_{B_r^\complement}\psi_0 \right\| _{\mathcal H} \leq \phi(r), \quad \lim_{r\to \infty}\phi(r)=0.
\end{equation}
\end{assumption}

We make a further truncation on the initial value $\psi_0$ by restricting it to only $B_r$, and denote $\Psi_0 := P_R \mathcal{X}_{B_r} \psi_0$.
 We now identify the truncated dynamics of wave functions in TBG as the following initial value problem of the finite dimensional Schr\"odinger equation
\begin{equation}
\label{eq:trunc_dynamic}
        i\hbar\partial_t \Psi = H_{R}\Psi, \quad \Psi(0) = \Psi_0  \in \ell^2(\Omega_R).
\end{equation}
The exact solution of the truncated Schr\"odinger equation is $\exp{(-iH_Rt)}\Psi_0 $.  We can establish the error of truncation as the difference between solutions of the infinite dimensional problem and the finite dimensional truncated problem. The essence of the estimates on the domain truncation error is a Combes-Thomas style estimate on the decay of the resolvent.

 \begin{theorem}[Exponential Decay of Resolvent]
 \label{thm:combes-thomas}
Let $H$ be a tight binding Hamiltonian that satisfies \cref{as:decay}. Fix $d$ positive and $\nu \in (0,1)$, then for any $z\in\mathbb C$ that satisfies $\operatorname{dist}(z,\sigma(H))\geq d$, there exists a constant $\alpha_{\max}$ depending on $h_0$, $\alpha_0$, $d$ and $\nu$ such that
\begin{equation}
\label{eq:CTestimate}
    \left|(z-H)^{-1}_{\vec R_i\sigma, \vec R'_j\sigma'}\right| \leq \frac{1}{\nu d}e^{-\alpha_{\max} \left|\vec R_i + \vec\tau_i^\sigma - \vec R'_j - \vec\tau_j^{\sigma'}\right| }.
\end{equation}
Recalling the distance between nearest graphene atoms $\delta$ \cref{eq:relative_shift} and tight-binding parameters $h_0$ and $\alpha_0$, 
the upper bound $\alpha_{\max}$ is the solution to the equation
\begin{equation}\label{eq:alpha_m_d_dependence}
	\frac{8\pi h_0 e^{\delta\alpha_0}}{|\Gamma|}\left[\frac{e^{\delta\alpha_{\max}}}{(\alpha_0 - \alpha_{\max})^2} - \frac{1}{\alpha_0^2} \right] = (1-\nu) d.
\end{equation}


\end{theorem}
\begin{proof}
    See \cref{sec:appendix_combes_thomas}. The proof uses a similar approach as previous results in \cite{Chen_Ortner_2016,E_Lu_2011}. We take advantage of the TBG structure to explicitly calculate the constants, and derive an exact dependence on the distance to the spectrum.
\end{proof}

We pause briefly to discuss the significance of the parameter $\nu$ appearing in Theorem \ref{thm:combes-thomas}. It can be seen from \eqref{eq:CTestimate} and \eqref{eq:alpha_m_d_dependence} that different choices optimize the short-range and asymptotic behavior of the estimates. If $\nu \rightarrow 0$, the constant on the right-hand side of \eqref{eq:CTestimate} tends to infinity. On the other hand, if $\nu \rightarrow 1$, it is clear from \eqref{eq:alpha_m_d_dependence} that $\alpha_{\max} \rightarrow 0$, so that the exponential decay in \eqref{eq:CTestimate} becomes trivial. It follows that the optimal choice of $\nu$ may depend on the situation. We set $\nu = \frac{1}{2}$ in what follows for simplicity of presentation.
We next provide an upper bound for the speed of propagation of wave-packets in TBG.
\begin{proposition}
\label{prop:speed_of_propagation}
    Consider the solution of the full Schr\"odinger equation with a discrete delta function at the origin $\vec{0}$, $\delta_{\vec{0}}$, as the initial condition. Letting 
 $d,\alpha_{\max}$ be defined as in \cref{thm:combes-thomas} and $\nu = 1/2$, and a contour $\gamma$ with $\operatorname{dist}(\gamma, \sigma(H)) > d$, then we have the estimate
    \begin{equation}
        \label{eq:shiftpeak}\left|\psi_{\vec R_i\sigma}(t)\right| = \left|\frac{1}{2\pi} \int_\gamma e^{-itz}(z-H)^{-1}\delta_{\vec{0}} \dee z\right| \leq \frac{C_\gamma}{\pi d}e^{- ( \alpha_{\max} \left|\vec R_i + \vec\tau_i^\sigma \right| - d t )},
    \end{equation}
    where $C_\gamma$ is the length of the contour.   
    For $d$ small, we have 
        \begin{equation} \label{eq:linearized}
         \left|\psi_{\vec R_i\sigma}(t)\right| \leq \frac{C_\gamma}{\pi d} e^{- \frac{d}{v_{\max}} \left(\left|\vec R_i + \vec\tau_i^\sigma \right| - v_{\max} t \right) + O(d^2)},\quad v_{\max} := \frac{16\pi  e^{\delta\alpha_0} (2+\delta\alpha_0) h_0}{|\Gamma| \alpha_0^3}.
    \end{equation}
 \end{proposition}
\begin{proof}
    Equation \cref{eq:alpha_m_d_dependence} is satisfied when $\alpha_{\max} = d = 0$, and using the implicit function theorem we can find a real-analytic function $d \mapsto \alpha_{\max}(d)$ with $\alpha_{\max}(0) = 0$ such that \cref{eq:alpha_m_d_dependence} is satisfied for all $d$ in a neighborhood of $0$. Inserting the power series of this function and balancing terms proportional to $d$ yields \cref{eq:linearized}.
\end{proof}
The constant $v_{\max}$ can be considered an upper bound for the speed of electron propagation in TBG. To see this, note that Proposition \ref{prop:speed_of_propagation} implies that, for compactly supported initial data $\psi_0$, the magnitude of $\psi_{\vec{R}_i \sigma}$ is exponentially small in $\text{dist}\left(\vec{R}_i + \vec{\tau}_i^\sigma,\supp \psi_0\right) - v_{\max} t$, where $\supp \psi_0$ denotes the support of $\psi_0$. It follows that the magnitude must be small until, at least, $t \approx \frac{\text{dist}\left(\vec{R}_i + \vec{\tau}_i^\sigma,\supp \psi_0\right)}{v_{\max}}$.  
The following theorem uses the finite speed of propagation to bound the truncation error.

\begin{theorem}[Truncation Estimate]\label{thm:trunc_main}
Suppose the TBG Hamiltonian satisfies \cref{as:decay}, and the wave-packet initial condition $\psi_0$ satisfies \cref{as:intial_condition}.
Let $\psi(t) : [0,\infty) \to \mathcal H$ be the solution to the initial value problem of the Schr\"odinger equation  on TBG \cref{eq:dynamic_hbar}
\begin{equation}
        i\hbar\partial_t \psi = H\psi, \quad \psi(0) = \psi_0.
\end{equation} 
Let $R$ be the truncation radius of the Hamiltonian, and $r$ be the truncation radius of the initial condition with $R > r$. Let $\Psi(t) : [0,\infty) \to \ell^2(\Omega_R)$ be the solution to the finite dimensional truncated equation.
\begin{equation}\label{eq:tb_dynamics_trunc}
        i\hbar\partial_t \Psi = H_{R}\Psi, \quad \Psi(0) = P_R \mathcal{X}_{B_r} \psi_0. 
\end{equation}
Let $\zeta(t)$ be the difference between the two solutions
 \begin{equation}
     \zeta(t) :=  \psi(t)- P_R^*\Psi(t),
 \end{equation}
and the truncation error be the norm $\|\zeta(t)\|_\mathcal{H}$.
For any spectral distance $d$, we can find a  closed contour $\Gamma$ around $\sigma(H)$ such that $\operatorname{dist}(\Gamma, \sigma(H)) > d$. Let $\nu$, $C_\gamma$ and $\alpha_{\max}$ be defined as in \cref{thm:combes-thomas}, 
we can bound the truncation error by
\begin{equation}
    \begin{split}
        &\|\zeta(t)\|_\mathcal{H} \leq \\ &\quad \sqrt{\frac{2}{\pi}} C_\gamma C_1(\alpha_{\max}, R) \left(\frac{h_0 |\Omega_R|}{\nu^2d^2} + \frac{1}{\nu d} \right) \frac{|\Omega_r|^{\frac{1}{2}}}{|\Gamma|^{\frac{1}{2}}} e^{-(\alpha_{\max}(R-r) - d t)}\left\| \mathcal{X}_{B_r} \psi_0 \right\|_{\mathcal H} + \phi(r), 
    \end{split}
\end{equation}
where $C_\gamma$ is the length of the contour, $|\Omega_R|$ is the number of orbitals in the truncated domain with radius $R$, $|\Gamma| = \sqrt{3}a^2/2$ is the area of a unit cell, and the coefficient $C_1(\alpha_{\max}, R)$ is explicitly
\begin{equation}
C_1(\alpha_{\max}, R) := e^{\delta\alpha_{\max}} \frac{(1+2R\alpha_{\max} - 2\delta\alpha_{\max})^{\frac 1 2}}{2\alpha_{\max}}.
\end{equation}
\end{theorem}

\begin{proof}
See \cref{sec:proof_main}.
\end{proof}

\begin{remark}
The truncation error decays exponentially for $R$ large. More importantly, suppose for $\epsilon_0 > 0$ and $T_0>0$, we can find $R_0$ such that for $R=R_0$ the truncation error $\|\zeta(t)\|_\mathcal{H} < \epsilon_0$  for any $t \in [0, T_0]$. 
We can use the exponent $dt - \alpha_{\max}(R-r)$ to conclude that if we want to control the error for a longer period of time, we only need to scale $R$ linearly. The scaling factor is directly related to the finite speed of propagation $v_{\max}$ introduced in \cref{prop:speed_of_propagation}. We are able to numerically verify the exponential decay, see \cref{fig:trunc_R}.
\end{remark}

\subsection{Bistritzer-MacDonald Hamiltonian}\label{sec:BM}

The Bistritzer-MacDonald (BM) model is a low-energy continuum approximation of TBG that predicts the physical properties of TBG with high accuracy \cite{Bistritzer_MacDonald_2011}. In particular, the model correctly predicted a series of magic angles, the largest being $\theta \approx 1.05^\circ$, where the TBG has Mott insulating and superconducting phases \cite{Cao_Fatemi_Demir_Fang_Tomarken_Luo_Sanchez-Yamagishi_Watanabe_Taniguchi_Kaxiras_etal_2018, Cao_Fatemi_Fang_Watanabe_Taniguchi_Kaxiras_Jarillo-Herrero_2018}. Three of the authors in this work  identified a parameter regime where the BM model approximates the tight-binding model of TBG with rigorous error estimate \cite{Watson_Kong_MacDonald_Luskin_2023}. In this subsection, we briefly introduce the BM model and how it can be related to the tight-binding dynamics of wave-packets.

We first define the momentum hops (with three-fold symmetry) as
\begin{equation}
    \begin{split}
        &\vec s_1 := \vec{K}_1 - \vec{K}_2 = |\Delta \vec{K}| \left( 0 , -1 \right)^\top, \\ &\vec{s}_2 := \vec{s}_1 + \vec{b}_{m,2} =  |\Delta \vec{K}| \left( \frac{\sqrt{3}}{2} , \frac{1}{2} \right)^\top,\\
        &\vec{s}_3 := \vec{s}_1 - \vec{b}_{m,1} = |\Delta \vec{K}| \left( - \frac{\sqrt{3}}{2} , \frac{1}{2} \right)^\top.
    \end{split}
\end{equation}
The length of the momentum hop vectors is the difference between Dirac points of the two monolayers, which depends on the twist angle $\theta$.
The momentum interlayer hopping matrices  are
\begin{equation}
        T_1 = \begin{pmatrix} 1 & 1 \\ 1 & 1 \end{pmatrix}, \quad 
        T_2 =  \begin{pmatrix} 1 & e^{- \frac{2\pi}{3}i } \\ e^{\frac{2\pi}{3}i} & 1 \end{pmatrix}, \quad 
        T_3 = \begin{pmatrix} 1 & e^{\frac{2\pi}{3}i } \\ e^{- \frac{2\pi}{3}i} & 1 \end{pmatrix}.
\end{equation}

The Bistritzer-MacDonald Hamiltonian $H_\text{BM}$ is an unbounded self-adjoint operator on the space $L^2(\mathbb{R}^2; \mathbb{C}^4)$ with domain $H^1(\mathbb{R}^2;\mathbb{C}^4)$. We introduce the Bistritzer-MacDonald Hamiltonian written out in physical units
\begin{equation} \label{eq:BM}
   H_\text{BM} := 
    \begin{pmatrix} v \vec{\sigma} \cdot (- i \nabla_{\vec r}) &  \displaystyle w \sum_{n = 1}^3 T_n e^{- i \vec{s}_n \cdot \vec r} \\  \displaystyle w \sum_{n = 1}^3 T_n^\dagger e^{i \vec{s}_n \cdot \vec{r}} & v \vec{\sigma}\cdot (- i \nabla_{\vec{r}}) \end{pmatrix},
    \end{equation}
    where $\vec{\sigma} = (\sigma_1, \sigma_2)^\top$ denotes the vector of Pauli matrices
\begin{equation}
        \sigma_1 = \begin{pmatrix} 0 & 1 \\ 1 & 0 \end{pmatrix}, \quad 
        \sigma_2 =  \begin{pmatrix} 0 & -i \\ i & 0 \end{pmatrix}.
\end{equation}
The parameters $v$ and $w$ control the strength of intralayer and interlayer hopping
\begin{equation}
    v := \hbar v_D, \quad w := \frac{ \hat{h}(\vec{K}; L) }{ |\Gamma| },
\end{equation} where $v_D$ is the Fermi velocity of monolayer graphene, and $\hat{h}(\vec{K}; L)$ is the two-dimensional Fourier transform of the hopping function evaluated at monolayer Dirac point $\vec{K}$ and interlayer distance $L$. Implicitly, the twist angle $\theta$ enters the BM Hamiltonian through the vectors $\vec{s}_n$. 

The matrix-valued scaled moir\'e interlayer potential 
\begin{equation}
T(\vec r) := w \sum_{n = 1}^3 T_n e^{- i \vec{s}_n \cdot\vec{r}}
\end{equation} is periodic up to a phase over the  moir\'e lattice vectors $\vec R_m$. We define a translation operator with a relative phase shift
\begin{equation}
    \tau_{\vec v} f(\vec r) := \operatorname{diag}(1, 1, e^{i\vec s_1 \cdot \vec v}, e^{i \vec s_1 \cdot \vec v}) f(\vec r + \vec v);
\end{equation}
by direct calculation we have
\begin{equation}
    H_{\text{BM}} \tau_{\vec R_m} = \tau_{\vec R_m} H_{\text{BM}}, \quad \vec R_{m} \in \Lambda_m.
\end{equation}
The commutative property allows us to express the spectrum of $H_\text{BM}$ as Bloch bands. Let $\vec k\in \Gamma_m^*$ be the wavenumber, the band structure can be calculated by solving the following eigenvalue problem 
\begin{equation}\label{eq:BM_eigen}
    H_{\text{BM}} \Phi(\vec r; \vec k) = E(\vec k) \Phi(\vec r; \vec k), \quad \Phi(\vec r; \vec k) := e^{i \vec k\vec r} \operatorname{diag}(1, 1, e^{i\vec s_1 \cdot \vec r}, e^{i \vec s_1 \cdot \vec r}) \phi(\vec r ; \vec k),
\end{equation}
where  $\phi(\vec r; \vec k) = \phi(\vec r + \vec R_m; \vec k)$ for any $\vec R_m \in \Lambda_m$.
For any $\vec k$, the problem can be efficiently solved by finding an orthogonal basis over the periodic cells, and using an eigenvalue solver to find the energies.

\begin{figure}[ht]
    \label{fig:moire_potential}
    \centering
    \includegraphics[width=.95\textwidth]{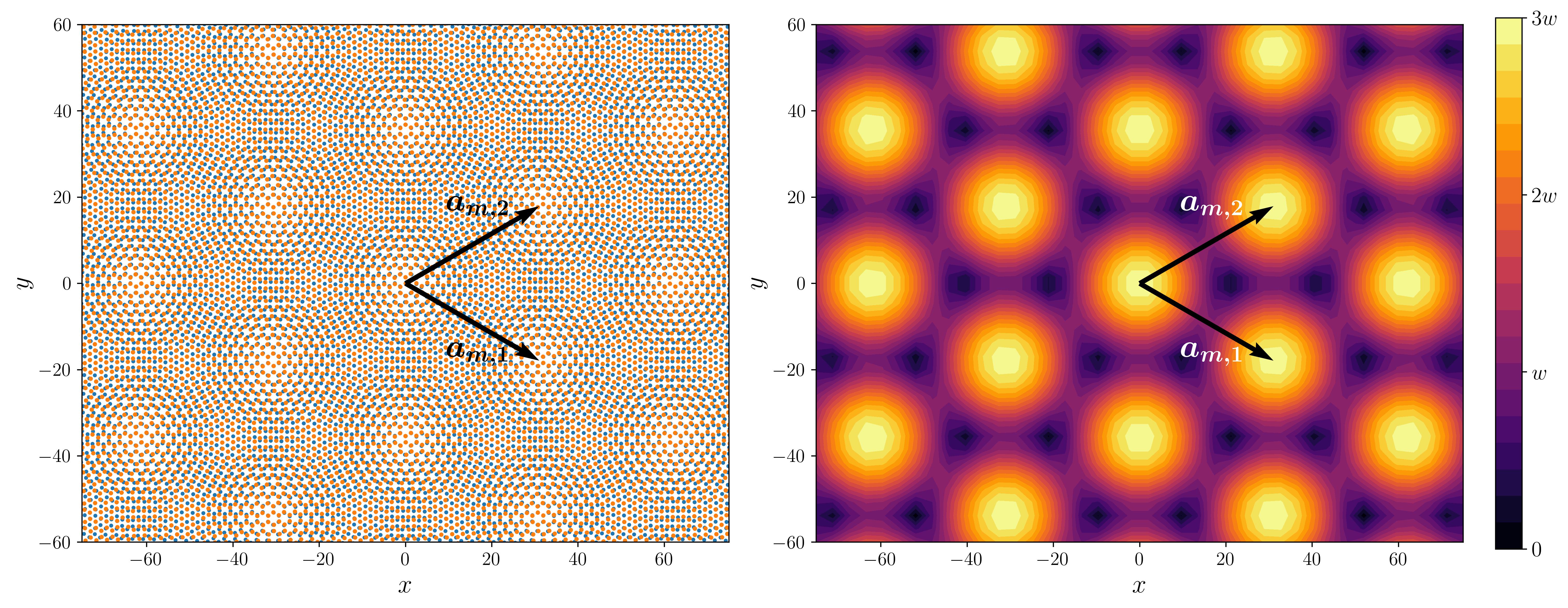}
    \caption{Left: The atomic structures of twisted bilayer graphene at twist angle $\theta=4^\circ$. The moir\'e lattice vectors are $\vec a_{m,1}$ and $\vec a_{m,2}$. Right: The modulus of one entry of the moir\'e interlayer potential $\left[T(\vec r)\right]_{11}$. It is periodic over the moir\'e lattice vectors.} 
\end{figure}

The BM Hamiltonian can be used to approximate the dynamics of wave-packets in TBG, in a specific parameter regime, in the following precise sense. 


\begin{theorem}[Approximation error of the BM model]\label{thm:BM}
Consider the tight-binding Hamiltonian $H_\text{TB}$ in \cref{ex:nearest_neighbor} and the BM Hamiltonian $H_\text{BM}$ in \cref{eq:BM}.

 Suppose there is a small dimensionless parameter $\epsilon > 0$ such that each component of the initial envelope function $f_0(\vec r) = \left(f_{1,0}^A(\vec r), f_{1,0}^B(\vec r), f_{2,0}^A(\vec r), f_{2,0}^B(\vec r)\right)^\top$ of the BM model satisfies the scaling relation
    \begin{equation}\label{eq:scaling}
        f_{i,0}^\sigma (\vec r) = \epsilon g_{i,0}^\sigma(\epsilon \vec r),
    \end{equation}
     where $g_i^\sigma$ has bounded eighth Sobolev norm  
    \begin{equation}\label{eq:regularity}
        \sup_{\sigma \in \{A,B\}} \sup_{i \in \{1,2\}} \| g_{i,0}^\sigma \|_{H^8(\mathbb{R}^2)} \leq C_{g_0},
    \end{equation}
    for some constant $C_{g_0}$.

    Further suppose the initial condition of tight-binding model is generated by 
        \begin{equation} \label{eq:wave_packet_0}
             (\psi_{0})_{\vec{R}_i\sigma} =  f^\sigma_{i,0}\left(  \vec{R}_i + \vec{\tau}^\sigma_i \right) e^{i \vec{K}_i \cdot (\vec{R}_i + \vec{\tau}_i^\sigma)}, \quad i \in \{1,2\}, \sigma \in \{A,B\}.
\end{equation}
Let $\psi(t)$ and $f(\vec r, t)$ be the solution to the time-dependent Schr\"odinger equations respectively
    \begin{equation}\label{eq:tb_dynamics}
    i\hbar\partial_t \psi = H_\text{TB}\psi, \quad \psi(0) = \psi_0,
\end{equation}
    \begin{equation}\label{eq:bm_dynamics}
    i\hbar\partial_t f = H_\text{BM}f, \quad f(\vec r, 0) = f_0(\vec r).
\end{equation}
Then $\psi(t)$ satisfies
\begin{equation}\label{eq:wave_packet_t}
             \psi_{\vec{R}_i\sigma}(t) =   f^\sigma_i(  \vec{R}_i + \vec{\tau}^\sigma_i, t)e^{i \vec{K}_i \cdot (\vec{R}_i + \vec{\tau}^\sigma_i)} + \eta_{\vec{R}_i\sigma}(t),       
    \end{equation}
where $i \in \{1,2\}, \sigma \in \{A,B\}$, and $\eta$ is the corrector.

The norm of the corrector depends on three small dimensionless parameters 
\[\epsilon,\ \theta,\text{ and }  \mathfrak{h} := aw/v = \frac{a \hat h(\vec K; L)}{ \hbar v_D|\Gamma|}.\]
Specifically, there exist constants $c_1,c_2>0$ which can be taken arbitrarily small, and a continuous function $\rho$ satisfying
\begin{equation}
	\lim_{\xi\to 0} \rho(\xi) = C, \quad \lim_{\xi\to \infty} \rho(\xi) = \infty,
\end{equation}
so that
    \begin{equation} \label{eq:key_estimate}
\| \eta(t) \|_\mathcal{H} \leq  \rho\left(\frac{\theta}{\epsilon}\right) \times \left( \left(\epsilon^2 + \theta\epsilon + \epsilon\mathfrak{h}^{1-c_1} + \mathfrak{h}^{2-c_2} \right) t\right).
\end{equation}

Under the additional assumption that there exist positive constants $\lambda_0$ and $\lambda_1$ such that
\begin{equation}\label{eq:BM_regime}
    \mathfrak{h}  = \lambda_0 \epsilon\quad\text{and}\quad \theta \leq \lambda_1 \epsilon,
\end{equation}
then there exists a constant $\epsilon_0 > 0$ such that for all $\epsilon < \epsilon_0$, and for any $c>0$, the leading order term is
\begin{equation} \label{eq:BM_error}
    \| \eta(t) \|_{\mathcal{H}} = O\left(\epsilon^{2 - c} t\right),
\end{equation}
where $c$ can be taken arbitrarily small.

\end{theorem}
\begin{proof}[Proof (sketch)]
    The detailed proof of this theorem as well as derivation of the BM model can be found in \cite{Watson_Kong_MacDonald_Luskin_2023}. We recover physical units for the BM time-propagation in order to match the tight-binding dynamics introduced in \cref{ex:nearest_neighbor}. We sketch the main ideas of the proof to explain the origins of the leading order terms. The estimate on the corrector relies on the estimate on the residual $r$, 
\begin{equation}
\begin{gathered}
    \| \eta(t)\|_\mathcal{H} \leq \int_0^t \|r(t')\|_\mathcal{H} \dee t', \\
     r_{\vec R_i\sigma}(t) := \epsilon \left[ H_{\text{BM}} g^\sigma_i \right] ( \epsilon (\vec{R}_i + \vec{\tau}^\sigma_i),  \epsilon t )
    e^{i \vec{K}_i \cdot (\vec{R}_i + \vec{\tau}^\sigma_i)} - H_{\text{TB}}\psi_{\vec R_i\sigma}(t).
\end{gathered}
\end{equation}
    Here, $g_i^\sigma$ are the components of $g$ that satisfies $ f(\vec r, t) = \epsilon g(\epsilon \vec r, \epsilon t)$. It solves a scaled IVP \cref{eq:bm_dynamics}, and, in the regime \cref{eq:BM_regime}, the Sobolev norms of $g_i^\sigma$ can be bounded in terms of those of the scaled initial data $g_{i,0}^\sigma$ \cref{eq:regularity} independently of all parameters (the function $\rho$ appearing in \cref{eq:key_estimate} is the constant depending on $\frac{\theta}{\epsilon}$ which appears in these estimates). 

We can write the residual as a sum over four terms $r = r^\text{I} + r^\text{II} + r^\text{III} + r^\text{IV}$, each term represents an approximation error, whose leading order and higher order terms can be estimated.

The first two terms originate from the monolayer interactions. The term $r^\text{I}$ is the second order and higher terms of the Taylor expansion of the monolayer Hamiltonian at the Dirac point, which captures the dispersion of the wave-packet.
    \begin{equation} \label{eq:r_I}
        \left\| r^{\text{I}}(t) \right\|_{\mathcal{H}} \lesssim \epsilon^2 \left( \sup_{i, \sigma}\left\| g_i^\sigma(\cdot,\epsilon t) \right\|_{H^2(\mathbb{R}^2)} + \epsilon^2 \sup_{i, \sigma}\left\| g_i^\sigma(\cdot,\epsilon t) \right\|_{H^6(\mathbb{R}^2)} \right).
    \end{equation}
The term $r^\text{II}$ is the result of using the untwisted Dirac operator on monolayers on rotated layers, therefore the error is dependent on the rotation angle $\theta$.
        \begin{equation} \label{eq:r_II}
            \left\| r^{\text{II}}(t) \right\|_{\mathcal{H}} \lesssim \theta \epsilon \left( \sup_{i, \sigma}\left\| g_i^\sigma(\cdot,\epsilon t) \right\|_{H^1(\mathbb{R}^2)} + \epsilon^2 \sup_{i, \sigma}\left\| g_i^\sigma(\cdot,\epsilon t) \right\|_{H^5(\mathbb{R}^2)} \right).
    \end{equation}

The next two terms originate from interlayer hopping. The term $r^\text{III}$ measures the ``local'' approximation \cite{Xie_MacDonald_2021}. It is the interaction of the wave-packet and the remainder of the Taylor expansion of $\hat h$ around the $\vec K$ point. For any $\mu_1 > 0$, we have the estimate
  \begin{equation} \label{eq:r_III}
    \begin{split}
        &\left\| {r}^\text{III}(t) \right\|_{\mathcal{H}} \lesssim   \\
        &\quad \epsilon \mathfrak{h} \left(  e^{\mu_1 L|\vec K|} \sup_{i, \sigma}\left\| g_i^\sigma(\cdot,\epsilon t) \right\|_{H^1(\mathbb{R}^2)} + \epsilon^{3/2} e^{L|\vec K|}\sup_{i, \sigma}\left\| g_i^\sigma(\cdot,\epsilon t) \right\|_{H^4(\mathbb{R}^2)} \right.  \\
        &\left. \vphantom{e^{\mu_1 L|\vec K|}}  \quad + \epsilon^2 e^{\mu_1 L|\vec K|}\sup_{i, \sigma}\left\| g_i^\sigma(\cdot,\epsilon t) \right\|_{H^5(\mathbb{R}^2)} +  \epsilon^{7/2} e^{L|\vec K|} \sup_{i, \sigma}\left\| g_i^\sigma(\cdot,\epsilon t) \right\|_{H^8(\mathbb{R}^2)} \right).
    \end{split}
    \end{equation}
Lastly, $r^\text{IV}$ captures the effect of hopping beyond nearest-neighbor in momentum space, which is excluded in the approximation. For any $\mu_2 > 0$, we have
    \begin{equation} \label{eq:r_IV}
    \begin{split}
        &\left\| {r}^{\text{IV}}(t) \right\|_{\ell^2(\mathbb{Z}^2;\mathbb{C}^2)} \lesssim  \\
        &\quad \mathfrak{h}^2 \left( e^{\mu_2 L|\vec K| } \sup_{i, \sigma}\left\| g_i^\sigma(\cdot,\epsilon t) \right\|_{L^2(\mathbb{R}^2)} + \epsilon^2 e^{\mu_2 L|\vec K| } \sup_{i, \sigma}\left\| g_i^\sigma(\cdot,\epsilon t) \right\|_{H^4(\mathbb{R}^2)} \right).
    \end{split}
\end{equation}

Recall that the Fourier transform of $h$ in \cref{eq:h_hat} gives
\begin{equation}
e^{-(1+\nu) L |\vec \xi|} \lesssim \hat h(\vec \xi; L) \lesssim e^{-L |\vec \xi|}, \quad \nu > 0,
\end{equation}
then we can rewrite the dependence on $L$ as dependence on $\hat h$, therefore on $\mathfrak{h}$. The leading order term estimate of the residual follows from the assumption that the Sobolev norms of $g_{i}^\sigma$ are bounded.  

Under the parameter regime given in \cref{eq:BM_regime}, the leading order term from 
\cref{eq:r_I,eq:r_II,eq:r_III,eq:r_IV} all balance. The overall error is $O(\epsilon^{2-c}t)$ for any $c>0$.
\end{proof}

\begin{remark}
The physical meanings of the fundamental parameters are as follows: $\theta$ is the twist angle in radians, $\mathfrak{h}$ is the ratio between interlayer hopping and intralayer hopping, and  $\epsilon$ separates the length scale of the wave-packet envelope and the plane wave parts (or equivalently it controls the concentration of wave-packets in momentum space). 

    Under the assumption that these parameters scale linearly, the error is at most $O(\epsilon^2 t)$. The scaling can also be stated in terms of length scales. Noting that the interlayer distance $L$ is related to the hopping function through the rough estimate $\hat h(\vec K; L) \sim e^{-|\vec K|L}$ , we can write the scaling rules as
\begin{equation}
    \frac{1}{\epsilon} \sim \frac{1}{\theta},\quad L \sim \ln{\frac{1}{\epsilon}},
\end{equation}
where $1/\epsilon$ is the length scale of the wave-packet envelope, and $1/\theta$ is the length scale of the moir\'e  lattice.
\end{remark}

\section{Numerical Simulations} 

\subsection{Convergence of domain truncation} \label{sec:num_trunc}
We test our error estimates using the nearest-neighbor tight-binding model $H_\text{TB}$ defined in \cref{ex:nearest_neighbor}. We use this specific model so we can compare its dynamics to that produced by the BM model directly. The matrix exponential for solving the truncated tight-binding model is calculated through Pad\'e approximation.

We can use the following set of parameters to  match the physical measurements of the monolayer graphene $\pi$-band energy, $v = 6.6 \text{ eV\AA},$ and the interlayer hopping energy, $w = 0.11 \text{ eV},$  given in the BM model in \cite{Bistritzer_MacDonald_2011}:

 \begin{equation} \label{eq:exp_regime}
         \theta = 1.05^\circ, \, a = 2.5\text{ \AA}, \, L = 3.5\text{ \AA}, \, t_0 = 3.048\text{ eV}, \, h_0 = 83.135\text{ eV}, \, \alpha_0 = 1\text{ \AA}^{-1}.
 \end{equation}
We choose $\theta$ at the magic angle, which is usually the most interesting angle for experiments.

To satisfy \cref{as:intial_condition}, we choose a normalized initial condition $\|\psi_0\|_\mathcal{H} = 1$. We will introduce the detailed construction of such initial conditions in \cref{sec:num_BM} that ensures exponential decay.  We set the truncation radius for the initial condition to be $r=10$. We compute the truncation error for a range of truncation radius $R$ and time $t$ in \cref{fig:trunc_R}.

\begin{figure}[ht]
    \label{fig:trunc_R}
    \centering
    \includegraphics[width=.55\textwidth]{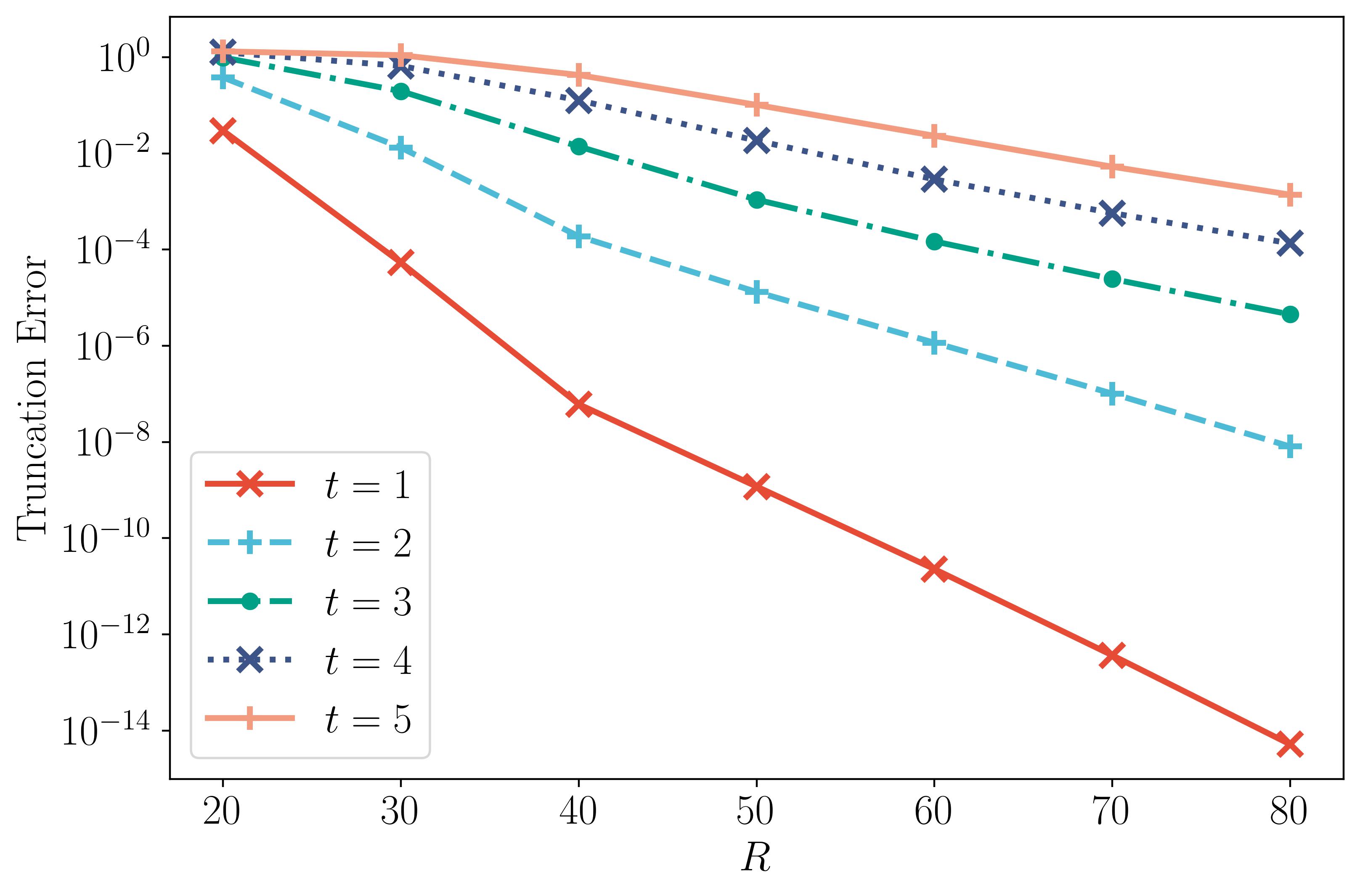}
    \caption{The relative error for the truncated tight-binding model in TBG. Each data point is the average of several simulations with different initial conditions while keeping their norm constant. As it is not possible to compute the infinite system $\psi(t)$ directly, we compare the solutions $P_R^*\Psi(t)$ to $P_{R'}^*\Psi(t)$, where the reference value is $R' = 86.60 \text{ \AA}$.
    When the truncation radius $R$ increases, $P_R^*\Psi(t)$ converges  exponentially.} 
\end{figure}

\subsection{Comparing tight-binding model to BM model}\label{sec:num_BM}
Since we have established the convergence of domain truncation for the tight-binding Hamiltonian, we can now compare the solutions of the truncated tight-binding model to the BM model. We prepare the initial condition according to \cref{thm:BM}, by first finding a suitable envelope function $f_0(\vec r) = \left(f_{1,0}^A(\vec r), f_{1,0}^B(\vec r), f_{2,0}^A(\vec r), f_{2,0}^B(\vec r)\right)^\top$ with bounded Sobolev norms. For a general Gaussian wave-packet envelope function, we can choose 
\begin{equation}\label{eq:gaussian_initial}
    f_0(\vec r) = \left(c_1^A, c_1^B, c_2^A, c_2^B\right)^\top G(\vec r), \quad G(\vec r) := e^{-\frac{|\vec r|^2}{2\sigma_r^2}},
\end{equation}
where $c_i^\sigma \in \mathbb C$ are the normalization coefficients for each component of the function. By setting $\sigma_r = \epsilon^{-1}$, we can control the wave-packet envelope length scale.

We can also utilize the band structure of the BM Hamiltonian to generate wave-packets with wavenumbers  concentrated in momentum space on any selected band. For any $\vec k_i \in \Gamma_m^*$, the energy on the $n$-th band is the $n$-th eigenvalue $E_n(\vec k_i)$ from \cref{eq:BM_eigen}, with the corresponding eigenfunction
$\Phi_n(\vec r; \vec k_i)$. We can set the initial wave-packet initial condition as the product of the eigenfunction and a two-dimensional Gaussian function
\begin{equation}\label{eq:k_initial}
    f_0(\vec r) = c\cdot \Phi_n(\vec r; \vec k_i)\cdot G(\vec r),
\end{equation}
where $c$ is the overall normalization coefficient, and $G$ is defined as in \cref{eq:gaussian_initial}. The group velocity for the wave-packet envelope is $\nabla_{\vec k} E_n(\vec k_i)$ \cite{Ashcroft_Mermin_1976}. A visual representation of the BM bands and two $\vec k_i$ points is presented in \cref{fig:BM_band_structure}.
\begin{figure}[ht] \label{fig:BM_band_structure}
    \centering
\includegraphics[width=.45\textwidth]{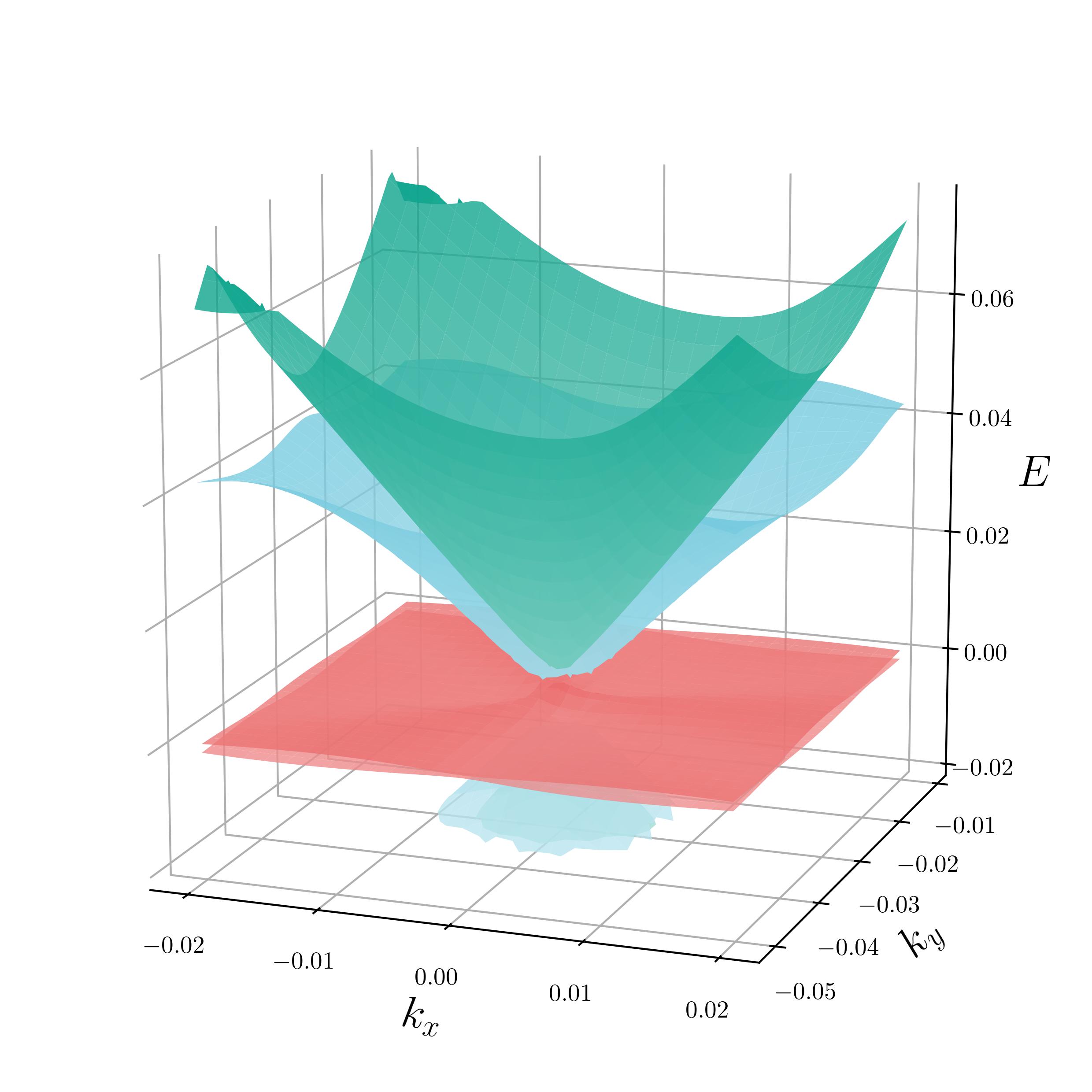}
\includegraphics[width=.35\textwidth]{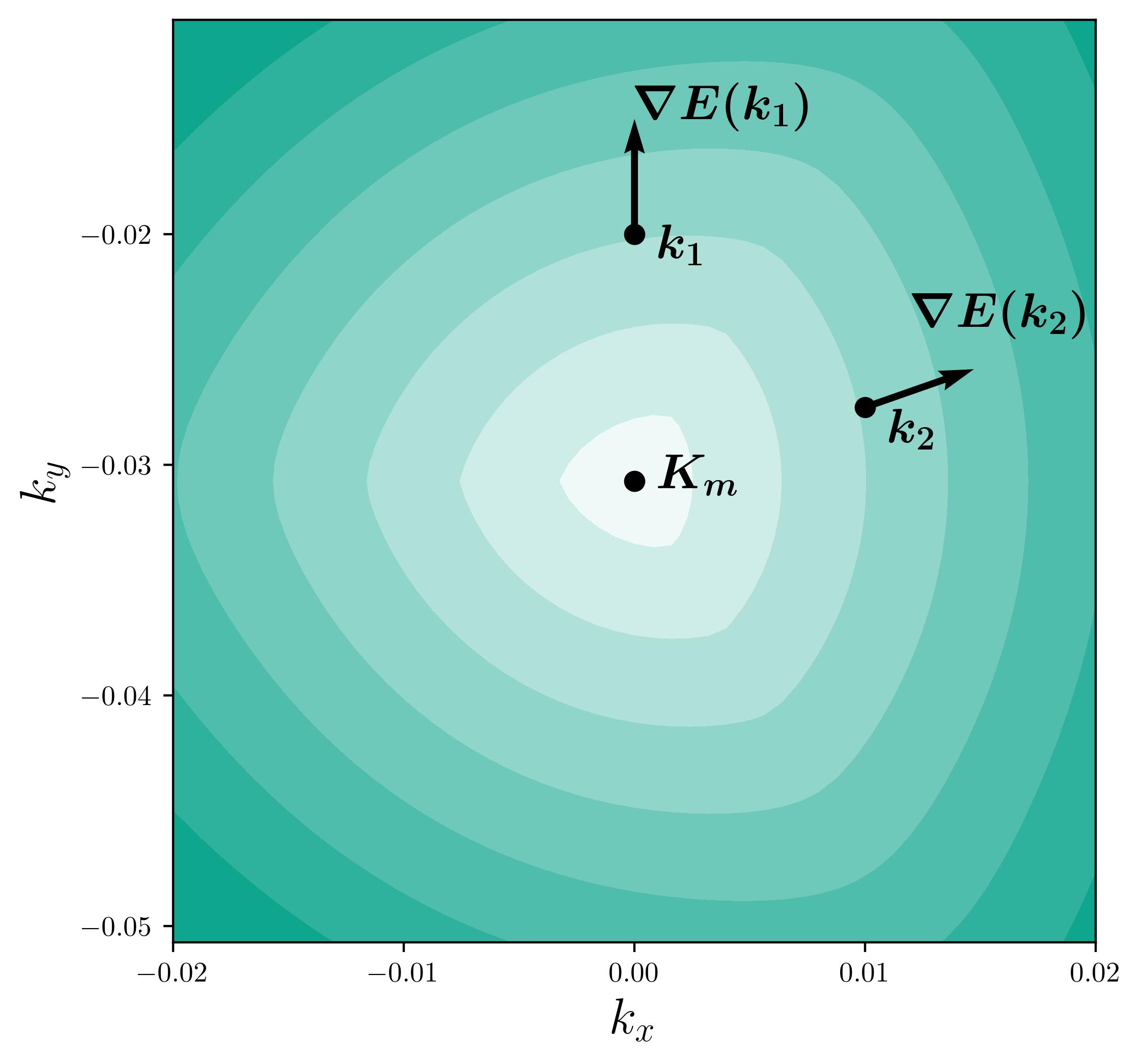}
\caption{The BM band structure of TBG at twist angle $\theta = 1.05 ^\circ$. $1.05^\circ$ is called a magic angle for TBG because there are a pair of almost flat moir\'e bands in the TBG band structure. Left: The flat bands (red) as well as two bands above the flat bands (blue and green) around a moir\'e $\vec K_m$ point. Right: The contour plot of the third band with two points $\vec k_1 = (0 , -0.02)^\top$ and $\vec k_2 = (0.01, -0.0275 )^\top$ in $k$-space, and an illustration of the gradient of energy $\vec\nabla E(\vec k_i)$ at these points.}
\end{figure}

For the tight-binding model, we use $f_0$ to generate the wave-packet initial condition $\psi_0$ through \cref{eq:wave_packet_0}. For suitable $R > r > 0$ we define the truncated initial condition as $\Psi_0 := P_R\mathcal{X}_{B_r}\psi_0$, and let it evolve according to \cref{eq:tb_dynamics_trunc}. We denote the solution to the truncated tight-binding model as $\Psi_{\text{TB}}(t)$.

For the BM model, we use the same initial condition $f_0$ to solve \cref{eq:bm_dynamics}, and the solution is $f(\vec r, T)$.  The solution is then mapped to a wave-function $\psi_{\text{BM}}(t)$ using \cref{eq:wave_packet_t}. Finally we project the wave-function to the truncated domain by letting $\Psi_{\text{BM}} := P_R \psi_{\text{BM}}$, so that we can compare it to the truncated tight-binding solution directly.

\cref{thm:trunc_main} and \cref{thm:BM} give a bound for the error between the truncated tight-binding model and the BM model
\begin{equation}
\left\|\Psi_{\text{BM}}(t) - \Psi_{\text{TB}}(t) \right\|_{\ell^2(\Omega_R)} \leq \| \zeta(t)\|_\mathcal{H} + \| \eta(t)\|_\mathcal{H}.
\end{equation}
From the convergence of the truncation error, for any $T>0$ we can choose the truncation radius $R$ sufficiently large so that $\|\eta(t)\|_\mathcal{H}$ is the dominant term for $t \leq T$. In this way, we are able to study the approximation error through the finite domain error $\left\|\Psi_{\text{BM}}(t) - \Psi_{\text{TB}}(t) \right\|_{\ell^2(\Omega_R)}$. For all numerical experiments, we use the truncated tight-binding model with $R = 86.60$ and $r = 10$. 

We present the approximation error for BM model with parameters  $\mathfrak{h} = aw/v \approx 0.042$, $\theta = 1.05^\circ \approx 0.017 \text{ rad}$, and $\epsilon = 0.1$. The $\epsilon$ value is carefully chosen such that the width of wave-packet envelope function is as large as possible while still contained in the truncated domain.  In \cref{fig:TB_BM_error}, we present the dynamics for wave-packets concentrated at $\vec k_1 = (0 , -0.02)^\top$ and $\vec k_2 = (0.01, -0.0275 )^\top$ on the third band. In \cref{fig:BM_flat_band}, we present the approximately zero group velocity results for a wave-packet initial condition concentrated at $\vec K_m$ point on the top flat band. We generate these initial conditions using \cref{eq:k_initial}.

\begin{figure}[ht] \label{fig:TB_BM_error}
    \centering
\includegraphics[width=.49\textwidth]{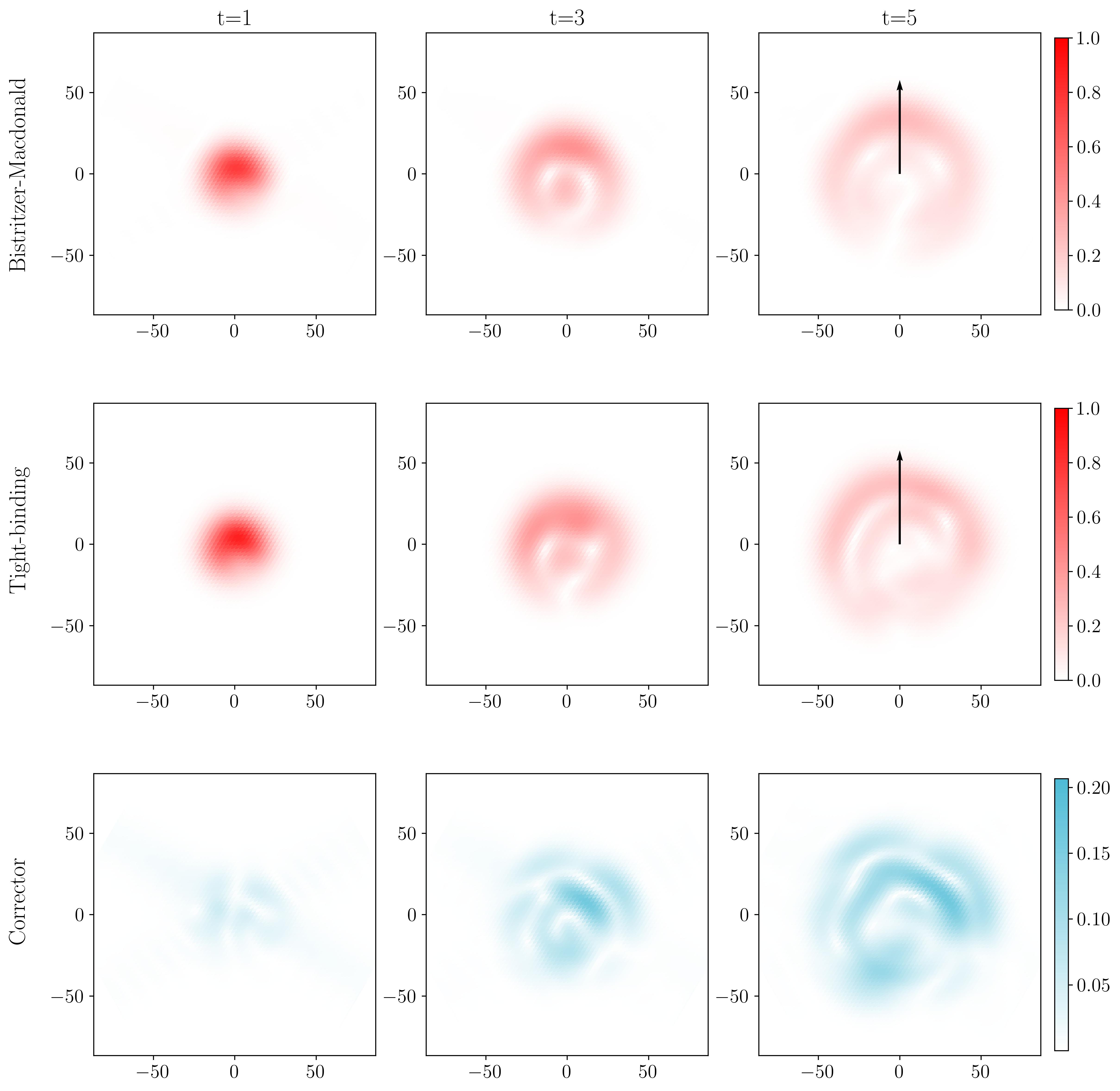}
\includegraphics[width=.49\textwidth]{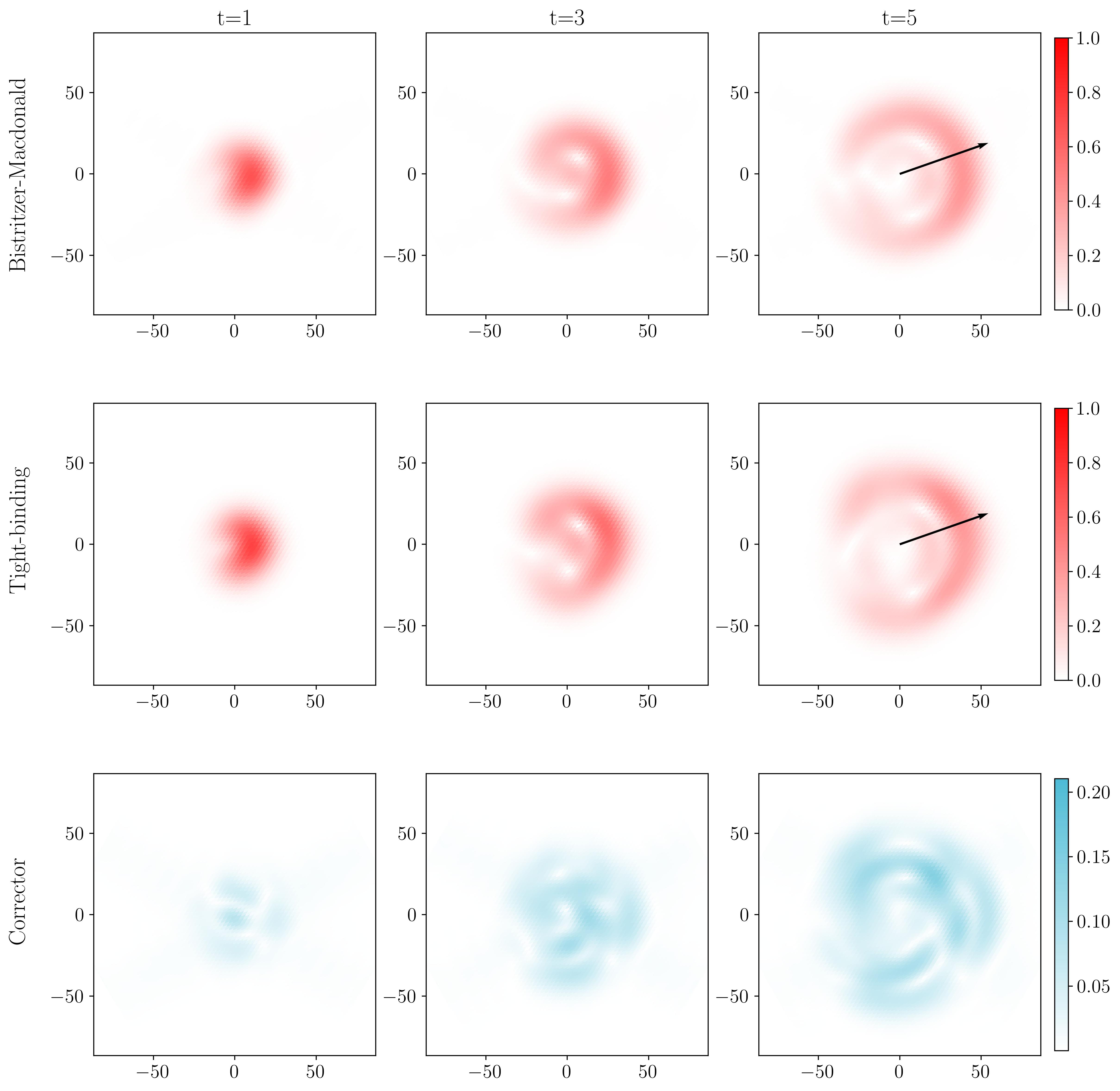}
    \caption{
    Left: The modulus of the wave-function for the BM model, the tight-binding model and the corrector for a wave-packet initial condition concentrated at $k_1$. Only one layer is presented, as the two layers have similar behaviour. The arrow represents the direction of $\vec\nabla E(\vec k_i)$.  Right: The same figure for a wave-packet concentrated at $k_2$. Recovering physical units, the axes have units $\AA$, and $t=T$ represents time at $ T \cdot\hbar \cdot\text{eV}^{-1} \approx T \times 6.6 \times 10^{-16} \text{s}$.
    }
\end{figure}

\begin{figure}[ht]
    \label{fig:BM_flat_band}
    \centering
    \includegraphics[width=.49\textwidth]{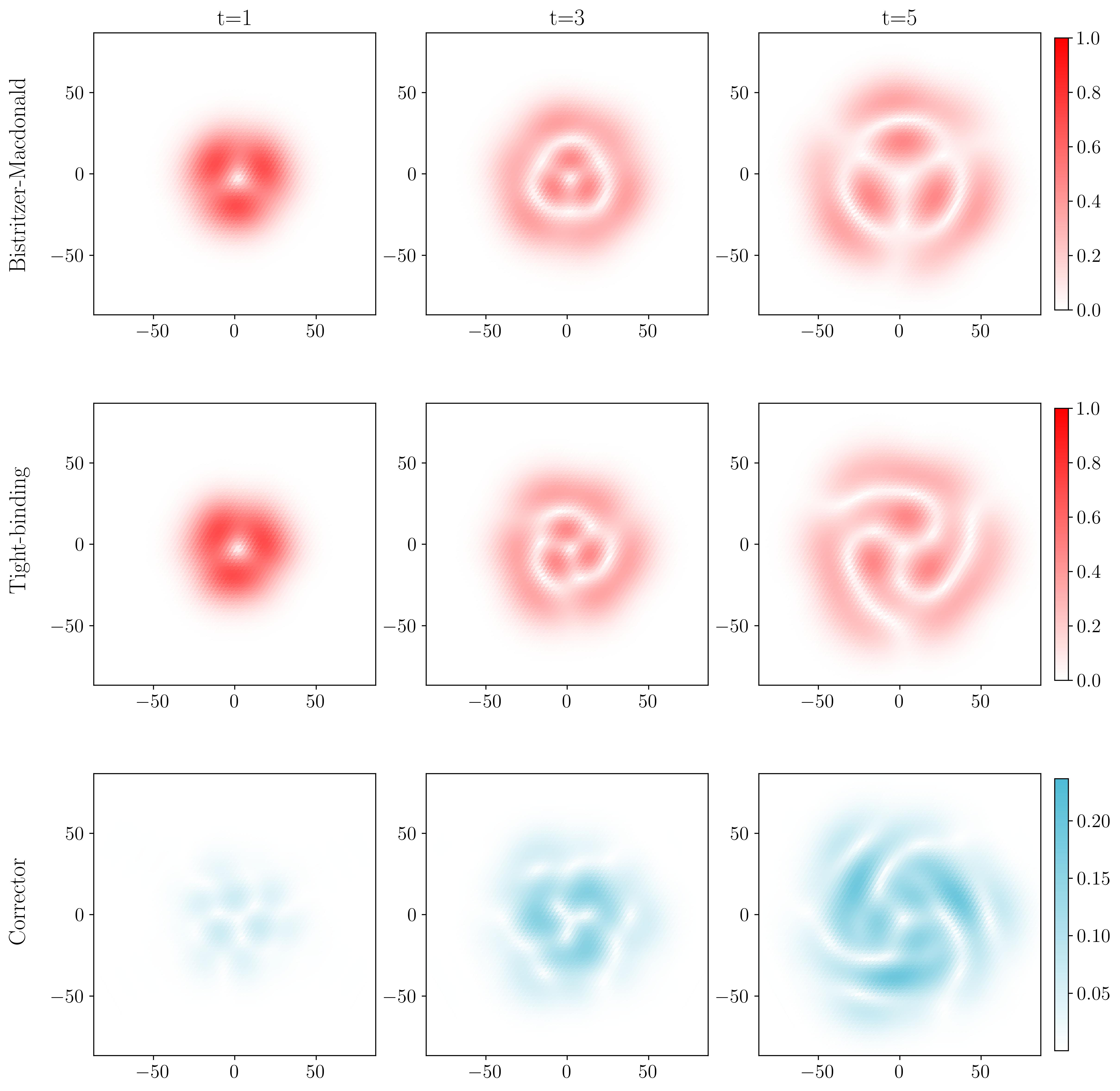}
        \includegraphics[width=.49\textwidth]{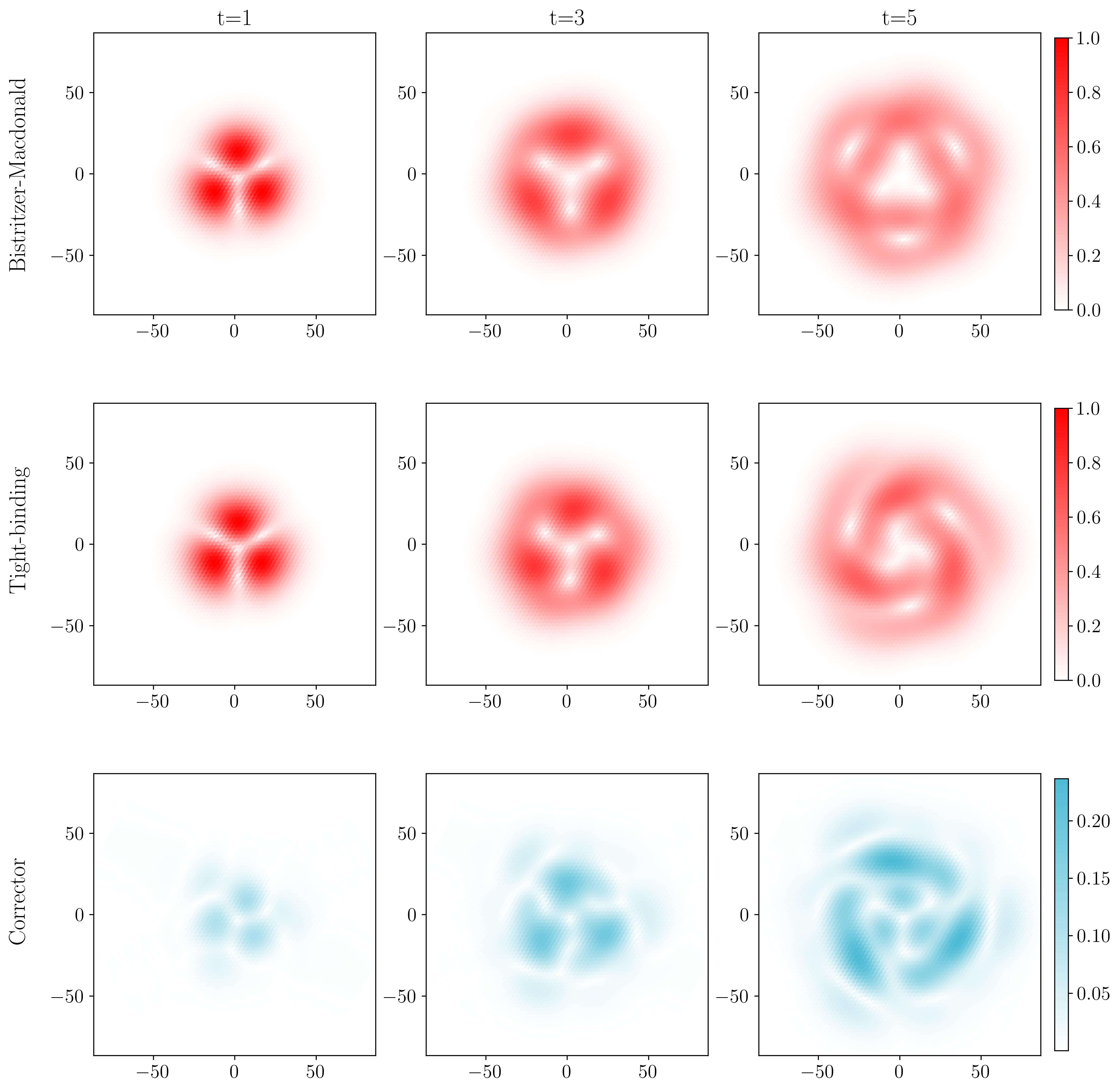}
    \caption{The dynamics of a wave-packet concentrated on a flat band at the degenerate point $\vec K_m$ of TBG. Both layer 1 (Left) and layer 2 (Right) are presented. We observe the approximately zero group velocity on both layers. The tight-binding dynamics show a twist-angle dependent rotation not seen in the BM dynamics.}
\end{figure}

\subsection{Sensitivity of parameters for the BM approximation}\label{sec:num_BM_beyond}
In this section, we discuss numerical experiments probing the sharpness of the estimates of Theorem \ref{thm:BM}. To do this, we numerically compute the approximation error $\| \eta(t)\|_\mathcal{H}$ between the BM model and tight-binding model as a function of time $t$ and the model parameters $\epsilon, \theta$, and $\mathfrak{h}$. We are interested, first, in whether the error scales like $\epsilon^2 t$ when the model parameters $\theta, \mathfrak{h}$ are scaled linearly with $\epsilon$ according to \eqref{eq:BM_regime}. We are also interested in whether the error scales like \eqref{eq:key_estimate} when we vary each of the model parameters $\epsilon, \theta, \mathfrak{h}$ individually while holding $t$ and the other model parameters fixed. Note that $\epsilon$ and $\theta$ can easily be varied individually because they appear in the BM model and tight-binding model directly. To vary $\mathfrak{h}$, we can change $h_0$ in the interlayer hopping function for the tight-binding model, and $w$ in the BM model.

To probe the error in the regime \eqref{eq:BM_regime} where $\theta, \mathfrak{h}$ scale linearly with $\epsilon$, we compute $\| \eta(t) \|_\mathcal{H}$ for various values of $t$ and $\epsilon$, while choosing $\mathfrak{h} = \lambda_0 \epsilon$ with $\lambda_0 = 0.42$ and $\theta = \lambda_1 \epsilon$ with $\lambda_1 = 0.17$. We make these choices of $\lambda_0, \lambda_1$ to be consistent with the estimated experimental values of each parameter in \cref{eq:exp_regime}. The results are shown in Figure \ref{fig:BM_epsilon_error}. We find that as we vary $t$ at fixed $\epsilon$, the slope of a linear fit to the error is $.92$, showing that the error is approximately linear in $t$ as expected. We find that as we vary $\epsilon$ at fixed $t$, the slope is $1.87$, showing that the error is approximately quadratic in $\epsilon$ as expected.

\begin{figure}[ht]
\label{fig:BM_epsilon_error}
    \centering
    \includegraphics[height=3.5cm]{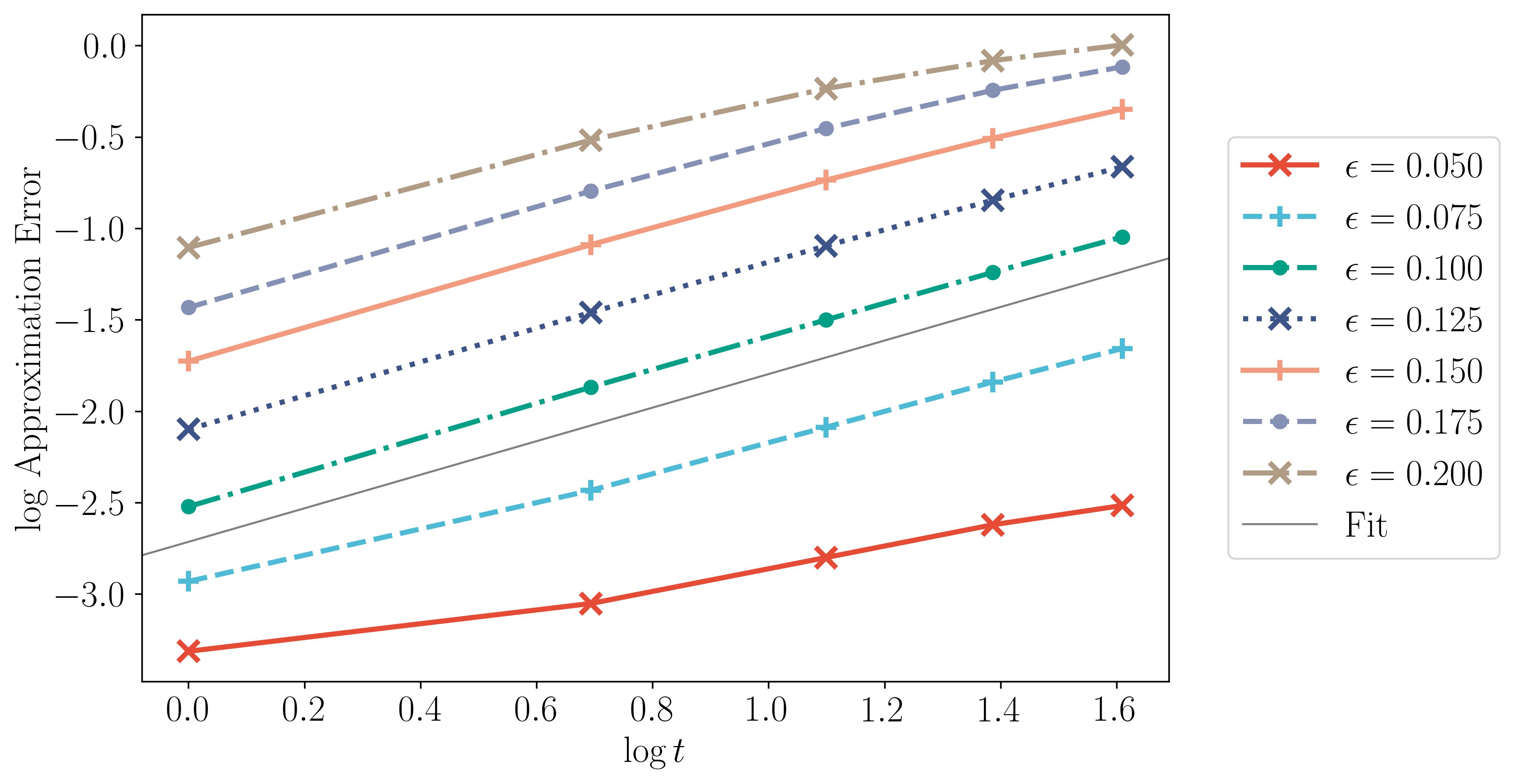}
    \includegraphics[height=3.5cm]{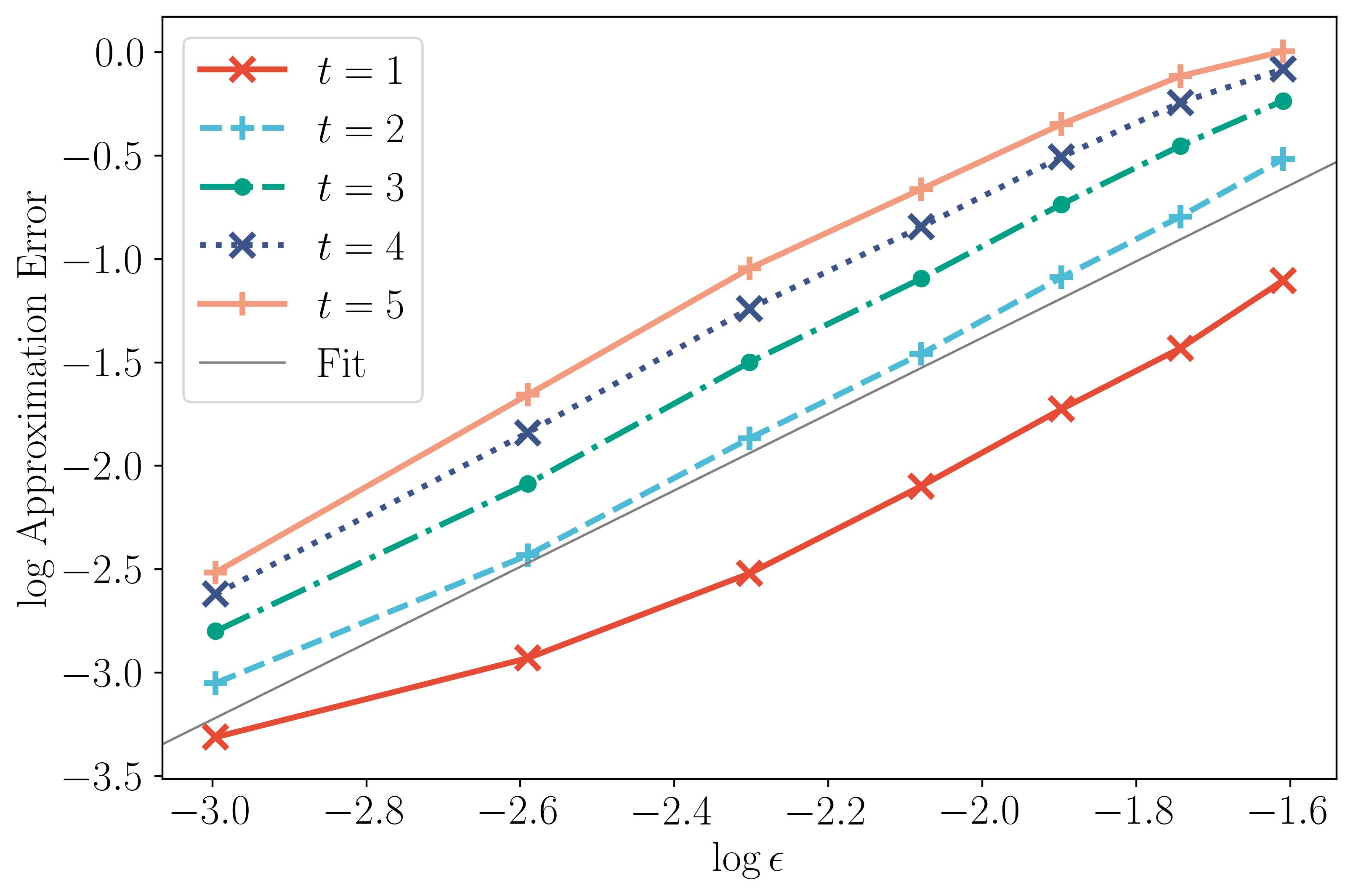}
    \caption{The approximation error between BM and tight-binding dynamics when the parameters $\epsilon$, $\theta$ and $\mathfrak{h}$ are in the regime given by \cref{eq:BM_regime}, presented in $\log-\log$ scale. Left: The error as a function of $t$ for various $\epsilon$. The slope of the linear fit is 0.92.
   Right: The same error data, but as a function of $\epsilon$. The slope of the linear fit is 1.87. These two figures verify that in the regime \cref{eq:BM_regime}, the BM approximation error is approximately $O(\epsilon^2 t)$, predicted in \cref{eq:BM_error}.
     } 
\end{figure}

To probe the dependence of the error on each model parameter individually, we again computed $\| \eta(t) \|_\mathcal{H}$, while varying each the model parameter while holding $t$ and the other model parameters fixed. When we increased $\mathfrak{h}$, while holding $\theta = 1.05^\circ$ and $\epsilon = 0.1$ fixed, for a range of $t$ values, we found the linear fit slope to be $0.72$. This is consistent with estimate \cref{eq:key_estimate}, according to which the dependence should be linear. When we increased $\epsilon$, while holding $\mathfrak{h} = 0.042$, $\theta = 1.05^\circ$, we found the linear fit slope to be $1.36$. This is consistent with estimate \cref{eq:key_estimate}, according to which the dependence should be quadratic. When we increased $\theta$, while keeping $\epsilon = 0.1$ and $\mathfrak{h} = 0.042$, the linear fit slope was only $0.08$. This surprising result suggests that the error is essentially independent of the twist angle up to $5^\circ$. This does not contradict estimate \cref{eq:key_estimate}, but suggests that it is not sharp. We aim to provide an analytical explanation of this phenomenon in future work.

\begin{figure}[ht]
\label{fig:BM_param_error}
    \centering
    \includegraphics[height=3.5cm]{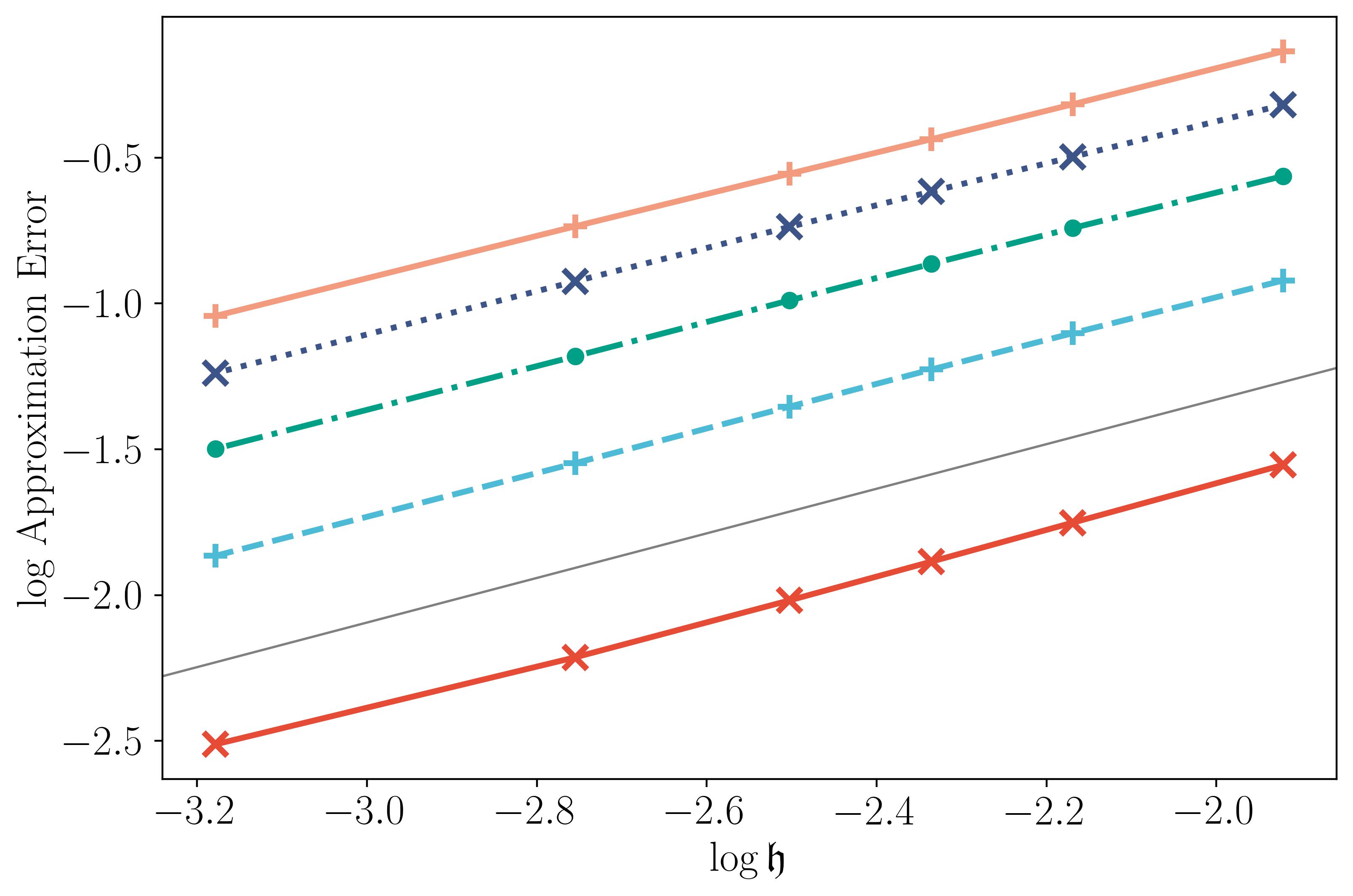}
    \includegraphics[height=3.5cm]{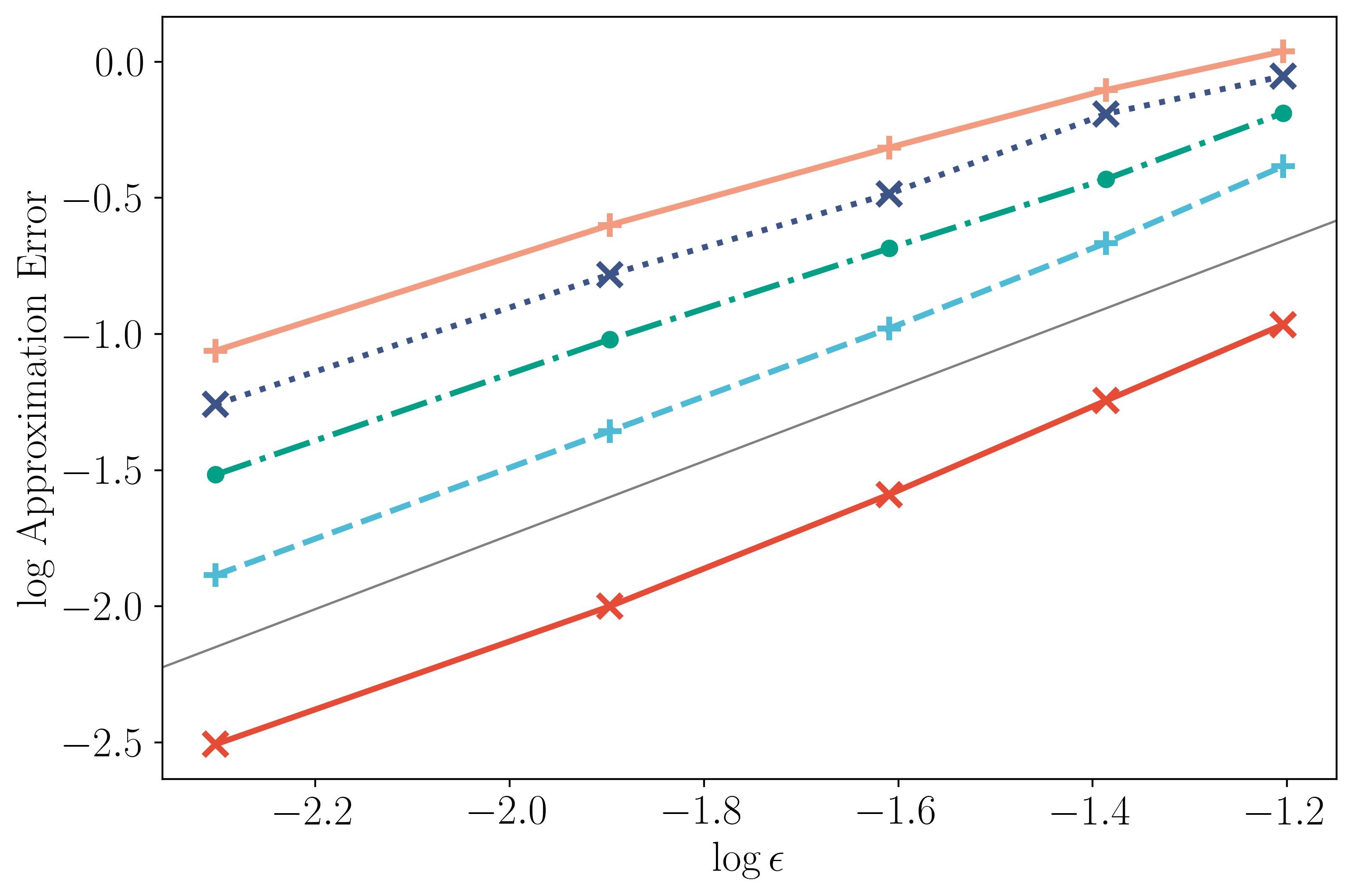}
\includegraphics[height=3.5cm]{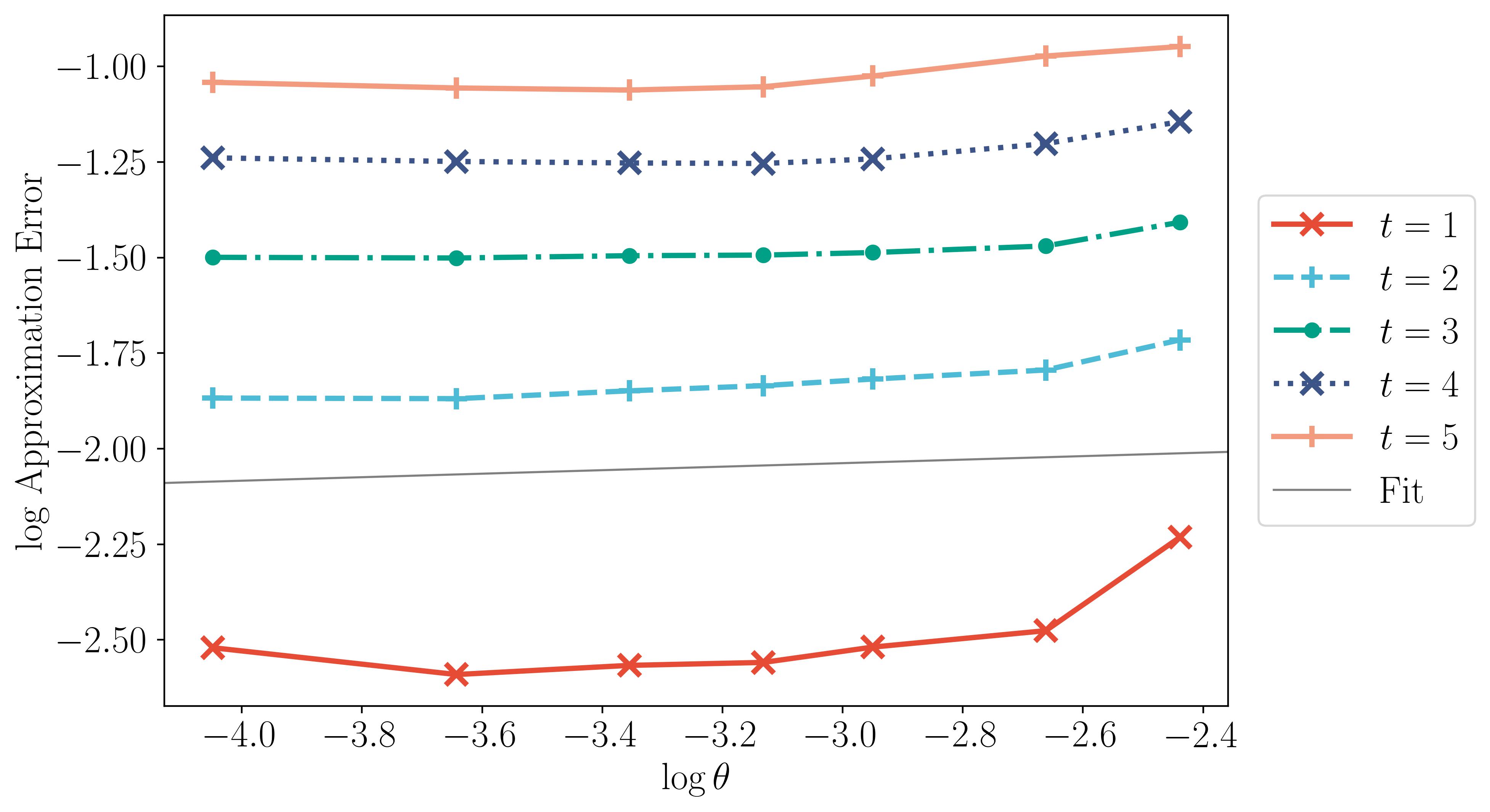}
   \caption{The approximation error between the BM and tight-binding models when one of the variables $\mathfrak{h}, \epsilon,\theta$ is changed, and the other two are constant. Each data point is the average of several numerical simulations with the same $\|\psi_0\|_\mathcal{H}$. The plots are presented in $\log-\log$ scale, and a linear fit is included to show the power relation on the parameters.
    Left: The error of increasing $\mathfrak{h}$, while keeping $\theta = 1.05^\circ$ and $\epsilon = 0.1$. The slope of the linear fit is 0.72.
   Right: The error of increasing $\epsilon$, while keeping $\mathfrak{h} = 0.042$ and $\theta = 1.05^\circ$. The slope of the linear fit is 1.36.
   Bottom: The error of increasing $\theta$, while keeping $\epsilon = 0.1$ and $\mathfrak{h} = 0.042$. The slope of the linear fit is 0.08, suggesting the approximation error does not depend on the twist angle $\theta$ up to $5 ^\circ$. } 
\end{figure}

\section{Conclusion}

In this paper, we considered the tight-binding model of an electron in twisted bilayer graphene. The accurate numerical computation of time evolved solutions of an electron is made challenging by the fact that the model is infinite dimensional and aperiodic at generic twist angles. We proposed approximating the dynamics by computations on finite domains. Using a speed of propagation estimate, we proved that the finite domain truncation error decays exponentially as the truncation radius increases.

Using this numerical method, we numerically investigated the range of validity of the effective PDE model of twisted bilayer graphene known as the Bistritzer-MacDonald model. We verified that in the regime \cref{eq:parameter_regime}, the Bistritzer-MacDonald approximation correctly captures the group velocity of spectrally concentrated wave-packet solutions of the tight-binding model. In particular, at the magic angle, we find wave-packet solutions with essentially zero group velocity \cite{Bistritzer_MacDonald_2011}. We also find that the main estimate on the approximation error \cref{eq:parameter_regime} from \cite{Watson_Kong_MacDonald_Luskin_2023} appears to be sharp, in the sense that it correctly captures the asymptotic dependence of the error on the parameter $\epsilon$ and time $t$. 

When we allow model parameters to vary independently, i.e., when we leave regime \cref{eq:parameter_regime}, we find that the more general estimate \cref{eq:main_estimate} usually captures the correct scaling of the error as a function of each parameter.  
The exception is when we vary the twist angle $\theta$ independently of other parameters. In this case, we find that the error grows very little when the twist angle is increased. This suggests that the Bistritzer-MacDonald approximation remains accurate even for a larger range of twist angles than predicted by the estimates of \cite{Watson_Kong_MacDonald_Luskin_2023}, as long as other model parameters are controlled.

In future work, we aim to provide an analytical explanation for the expanded range of validity of the Bistritzer-MacDonald model found here. We will also investigate efficient numerical methods for computing dynamics in incommensurate heterostructures along the lines of \cite{Wang_Chen_Zhou_Zhou_2021}, and numerically investigate the ranges of validity of ``corrected'' Bistritzer-MacDonald models which have appeared in the physics literature, e.g. \cite{VafekKang2023,PhysRevResearch.1.013001,Balents2019}.

\bibliographystyle{siamplain}
\bibliography{references}

\newpage
\appendix

\section{Proof of \cref{lem:bounded}}
\label{prf:h_bound}
We bound the operator norm of $H$ using interpolation. We first claim that $H$ is a bounded operator from $\ell^\infty(\Omega)$ to $\ell^\infty(\Omega)$, and the operator norm $\|H\|_\infty$ can be computed directly by 
\begin{equation}
\|H\|_{\infty} = \sup_{\vec R'_j\sigma \in \Omega} \sum_{\vec R_i\sigma \in \Omega} \left|H_{\vec R_i\sigma, \vec R'_j\sigma'} \right|.
\end{equation}
Fix any $\vec R'_j\sigma' \in\Omega$, we have the estimate
    \begin{equation}
    \label{eq:bound1}
    \begin{split}
    \sum_{\vec R_i\sigma \in \Omega} \left|H_{\vec R_i\sigma, \vec R'_j\sigma'} \right| & \leq\sum_{\vec R_i\sigma \in \Omega} h_0 e^{-\alpha_0 \left| \vec R_i + \vec \tau_i^\sigma - \vec R'_j - \vec \tau_j^{\sigma'}\right|} \\
    & = h_0 \sum_{i \in \{1, 2\}}\sum_{\sigma \in \{A, B\}}\sum_{\vec R_i\in \mathcal{R}_i} e^{-\alpha_0 \left| \vec R_i + \vec \tau_i^\sigma - \vec R'_j - \vec \tau_j^{\sigma'}\right|}     
    \end{split}
    \end{equation}

For fixed $i,\sigma$, we can rewrite the last summation of \cref{eq:bound1} as a summation over lattice points
\begin{equation}
	\sum_{\vec R_i\in \mathcal{R}_i} e^{-\alpha_0 \left| \vec R_i - \vec \tau \right|},\quad \vec \tau:=  \vec R'_j + \vec \tau_j^{\sigma'} - \vec \tau_i^\sigma,
\end{equation}
and we bound the summation using numerical integration techniques. We fix $\Gamma$ as the Wigner-Seitz (hexagonal) unit cell with $\vec{R} = 0$ at its center (see \cref{fig:lattice_sum}), so that
we can cover $\mathbb{R}^2$ with unit cells of the form $\Gamma_{\vec R} := \vec R + \Gamma$, $\vec R \in \mathcal{R}_i$. Within each unit cell we have the inequality 
\begin{equation}
	\left|\Gamma_{\vec R}\right| \min_{\vec x \in\Gamma_{\vec{R}}} e^{-\alpha_0 \left| \vec x - \vec \tau \right|} \leq \int_{\Gamma_{\vec R}} e^{-\alpha_0 \left| \vec x- \vec \tau \right|} \dee \vec x.
\end{equation}
The minimum are achieved on the boundary of $\Gamma_{\vec R}$ by convexity. 
Suppose the minimum is achieved at $\vec x = \vec Q$, from triangular inequality we have
\begin{equation}
\begin{split}
	e^{-\alpha_0 \left| \vec R - \vec \tau \right|} \leq e^{-\alpha_0 \left( \left| \vec Q - \vec \tau \right| - \left| \vec R - \vec Q \right| \right)} \leq \frac{e^{\alpha_0 |\vec R - \vec Q |}}{\left|\Gamma_{\vec R}\right|} \int_{\Gamma_{\vec R}} e^{-\alpha_0 \left| \vec x - \vec \tau \right|} \dee \vec x,
\end{split}
		\end{equation}
$|\Gamma_{\vec{R}}| = |\Gamma| = \sqrt{3}a^2/2$. We also have 
\begin{equation}
	\sup_{\vec Q \in \partial \Gamma_{\vec{R}}} | \vec Q - \vec R| \leq \frac{a}{\sqrt{3}} = \delta,
\end{equation}
so we can bound the summation over lattices using integration over $\mathbb R^2$,
\begin{equation} \label{eq:lattice_sum_bound}
		 \sum_{\vec R_i\in \mathcal{R}_i} e^{-\alpha_0 \left| \vec R_i - \vec \tau \right|} 
		\leq  \frac{ e^{\delta\alpha_0}}{|\Gamma|}\int_{\mathbb{R}^2} e^{-\alpha_0 \left| \vec x - \vec \tau\right|} \dee \vec x.
\end{equation}
 
\begin{figure}[h]\label{fig:lattice_sum}
\centering
\includegraphics[width=.5\textwidth]{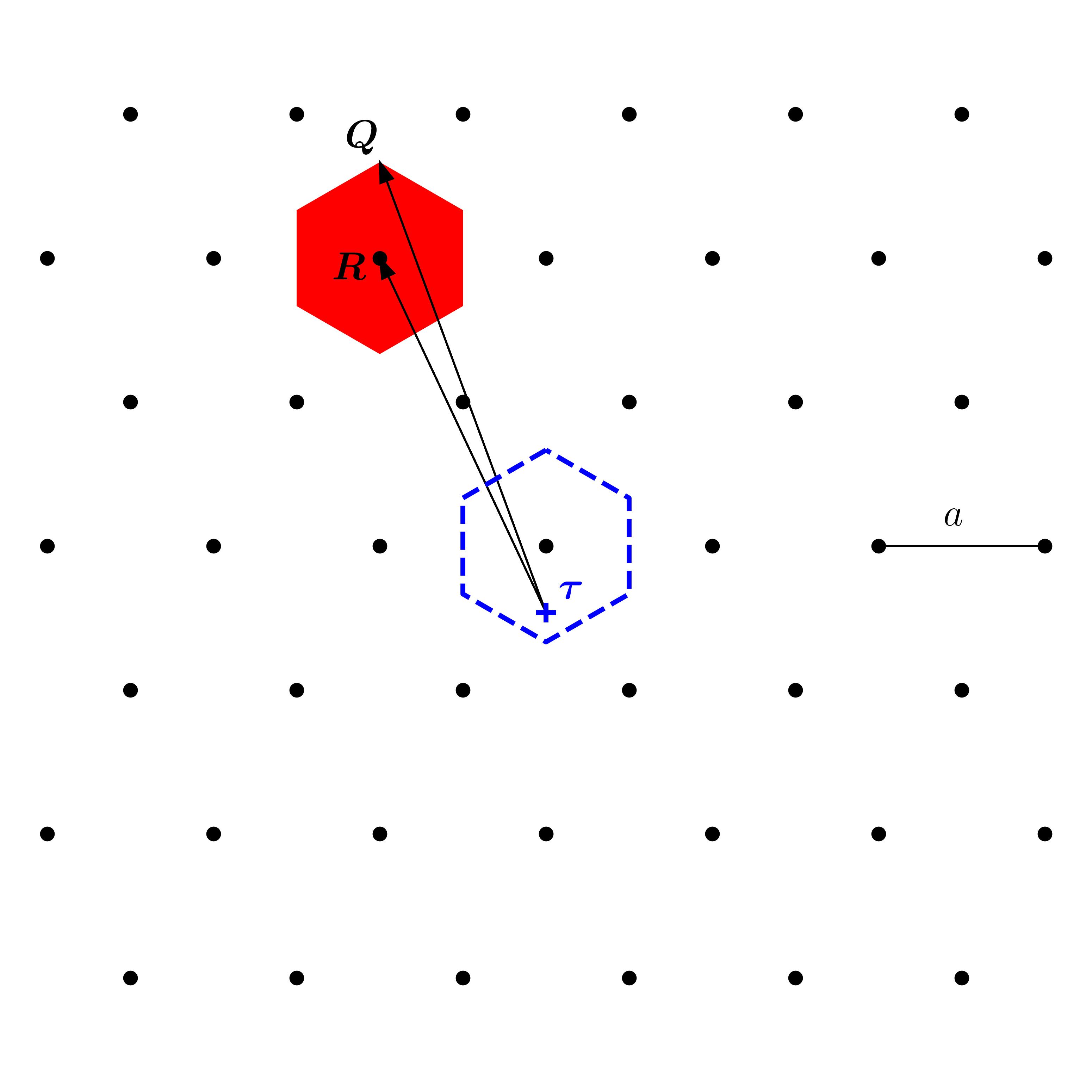}
\caption{The Bravais lattice (black) and a Wigner-Seitz unit cell $\Gamma_{\vec{R}}$ of $\vec R$ (red). After a shift of the lattice we can always ensure $\vec \tau $ is in a unit cell that contains the origin (blue).}
\end{figure}

Notice that we can use a change of variables to eliminate $\vec \tau$, so we have a uniform bound that is independent of the layer and sublattice indices. We can sum over two layer and two sublattices to get an upper bound
\begin{equation}
	\| H \|_\infty \leq \frac{4 h_0 e^{\delta\alpha_0} }{|\Gamma|} \int_{\mathbb{R}^2}e^{-\alpha_0 |\vec x|} \dee \vec x = \frac{8\pi h_0 e^{\delta\alpha_0}  }{|\Gamma| \alpha_0^2}.
\end{equation}

Since $H$ is self-adjoint, $H$ is also a bounded operator from $\ell^1(\Omega)$ to $\ell^1(\Omega)$ with the same bound on the operator norm.
\begin{equation}
\|H\|_1 = \sup_{\vec R_i\sigma \in \Omega} \sum_{\vec R'_j\sigma' \in \Omega} \left|H_{\vec R_i\sigma, \vec R'_j\sigma'} \right| = \sup_{\vec R'_j\sigma \in \Omega} \sum_{\vec R_i\sigma \in \Omega} \left|H_{\vec R_i\sigma, \vec R'_j\sigma'} \right| = \|H\|_\infty.
\end{equation}

The operator norm on $\mathcal{H}$ can be bounded using Riesz-Thorin theorem 
\begin{equation}
\|H\|_2 \leq \|H\|_\infty^{\frac{1}{2}} \|H\|_1^{\frac{1}{2}} \leq \frac{8\pi h_0 e^{\delta\alpha_0 } }{|\Gamma| \alpha_0^2}.
\end{equation}

\
\section{Proof of \cref{thm:combes-thomas}}
\label{sec:appendix_combes_thomas}
For ease of notation we use $x:= \vec R_i\sigma$, $y:=\vec R'_j\sigma'$ to represent the indices of TBG orbitals, and $\vec x := \vec R_i + \vec \tau_i^\sigma$, $\vec y := \vec R'_j + \vec \tau_j^{\sigma'}$ to represent their respective physical locations. Notice we are able to simplify this notation because the proof does not rely on the specific layer and sublattice structures of TBG.

Fix index $k \in \Omega$ with physical location $\vec k$, we define a bounded linear operator $B_\alpha$ for any $\alpha \geq 0$ on $\mathcal{H}$
\begin{equation}
\left(B_\alpha\right)_{x y} := 
\begin{cases}
    e^{\alpha |\vec x - \vec k|},  \; \text{if } x = y \\
    0, \; \text{otherwise}.
\end{cases}
\end{equation}

The entries of the operator $B_\alpha HB_\alpha^{-1}-H$ are, explicitly,
\begin{equation}
    \begin{split}
        \left(B_\alpha HB_\alpha^{-1} - H\right)_{xy} &= e^{\alpha |\vec x - \vec k|} H_{xy} e^{-\alpha |\vec y - \vec k|} -H_{xy}  \\
        & = H_{xy} \left[e^{\alpha ( |\vec x - \vec k|-|\vec y - \vec k|)} - 1\right]. \\
    \end{split}
\end{equation}
Using similar arguments as \cref{prf:h_bound}, we can bound the operator norm through a summation over the lattice, and then bound that by an integral. First, note that
\begin{equation}
    \begin{split}
        \left \| B_\alpha HB_\alpha^{-1}-H \right \|_\infty & \leq \sup_{x\in \Omega} \sum_{y\in \Omega} |H_{xy}|(e^{\alpha|\vec x - \vec y|}-1)\\
     & \leq h_0 \sup_{x \in \Omega} \sum_{y\in \Omega} e^{-\alpha_0|\vec x - \vec y|}\left( e^{\alpha|\vec x - \vec y|} - 1 \right).
     \end{split}
     \end{equation}
It is straightforward to bound the same operator in the $\ell^1$ norm by the same quantity. So, applying Riesz-Thorin, we have
\begin{equation}\label{eq:BHB}
        \left \| B_\alpha HB_\alpha^{-1}-H \right \| \leq h_0 \sup_{x \in \Omega} \sum_{y\in \Omega} e^{-\alpha_0|\vec x - \vec y|}\left( e^{\alpha|\vec x - \vec y|} - 1 \right).
     \end{equation}

 We then bound the summation for any fixed $i, \sigma$. Denote $\tau := \vec R'_j + \vec \tau_j^{\sigma'} - \tau_i^\sigma$,  and let $\Gamma_{\vec R}$ be the Wigner-Seitz unit cell associated with $\vec R$ as in \cref{prf:h_bound}, and $\vec Q_{\vec R} \in \Gamma_{\vec R}$, we have the estimates using triangular inequality
\begin{equation}
	e^{-\alpha_0|\vec R - \vec \tau|} \leq e^{\delta\alpha_0 }e^{-\alpha_0|\vec Q_{\vec R} - \vec \tau|}, \quad e^{\alpha|\vec R - \vec \tau|} - 1 \leq  e^{\delta\alpha}e^{\alpha |\vec Q_{\vec R} - \vec \tau|}	- 1.
\end{equation}
Summing over the lattice $\mathcal{R}_i$, and choosing $\vec Q_{\vec R_i}$ such that \begin{equation}
e^{-\alpha_0|\vec Q_{\vec R_i} - \vec \tau|}\left( e^{\delta\alpha}e^{\alpha |\vec Q_{\vec R_i} - \vec \tau|}	- 1 \right) \end{equation}
is minimized in each unit cell, we then have
\begin{equation}
	 \sum_{\vec R_i\in \mathcal{R}_i} e^{-\alpha_0 \left| \vec R_i - \vec \tau \right|}\left( e^{\alpha|\vec R_i - \vec \tau|} - 1 \right)
		\leq  \frac{ e^{\delta\alpha_0}}{|\Gamma|}\int_{\mathbb{R}^2} e^{-\alpha_0 \left| \vec x - \vec \tau\right|} \left(e^{\alpha \delta}e^{\alpha|\vec x - \vec \tau|} - 1 \right) \dee \vec x.
\end{equation}

Multiplying by number of layers and sublattices, and evaluating the integral, we conclude that the integral converges only when $\alpha < \alpha_0$, and \cref{eq:BHB} is bounded by
\begin{equation} \label{eq:BHB_bound}
\frac{8\pi h_0 e^{\delta\alpha_0}}{|\Gamma|}\left[\frac{e^{\delta\alpha}}{(\alpha_0 - \alpha)^2} - \frac{1}{\alpha_0^2} \right].
\end{equation}

It is clear that \cref{eq:BHB_bound} is an increasing function of $\alpha$ for $0 \leq \alpha < \alpha_0$ which equals $0$ at $\alpha = 0$ and $\rightarrow \infty$ as $\alpha \rightarrow \alpha_0$. Thus, for any $\nu \in (0,1)$, we can define $\alpha_{\max}$ such that 
\begin{equation}
	\frac{8\pi h_0 e^{\delta\alpha_0}}{|\Gamma|}\left[\frac{e^{\delta\alpha_{\max}}}{(\alpha_0 - \alpha_{\max})^2} - \frac{1}{\alpha_0^2} \right] = (1-\nu) d.
\end{equation}
We then have that $\left\|B_\alpha H B_\alpha^{-1}-H\right\| \leq (1-\nu)d$ for all $\alpha \leq \alpha_{\max}$. 

Now, notice that 
\begin{equation}
    \begin{split}
        B_\alpha(z-H)^{-1}B_\alpha^{-1} & = \left(z - B_\alpha HB_\alpha^{-1}\right)^{-1} \\
        & = \left(z - H + H - B_\alpha HB_\alpha^{-1}\right)^{-1} \\
        & = (z-H)^{-1} \left[I - (B_\alpha HB_\alpha^{-1}-H)(z-H)^{-1}\right]^{-1}.
    \end{split}
\end{equation}
The assumption $\operatorname{dist}(z, \sigma(H))\geq d$ gives $\left\|(z-H)^{-1} \right \|\leq d^{-1}$. Our choice of $\alpha$ ensures that the operator $z-B_\alpha HB_\alpha^{-1}$ is invertible and 
\begin{equation}
   \left\|B_\alpha(z-H)^{-1}B_\alpha^{-1} \right\| \leq \frac{1}{d} \left( 1 - (1-\nu)d\cdot\frac{1}{d} \right)^{-1} = \frac{1}{\nu d}.
\end{equation}
Moreover, 
\begin{equation}
    \begin{split}
        \left|[B_\alpha(z-H)^{-1}B_\alpha^{-1}]_{xy}\right| &= \left|[(z-H)^{-1}_{xy} e^{\alpha (|\vec x - \vec k|-|\vec k - \vec y|)}\right|\\
        &\leq \left \| B_\alpha(z-H)^{-1}B_\alpha^{-1} \right \| \leq\frac{1}{\nu d},
    \end{split}
\end{equation} 
which gives
\begin{equation}
    \left|[(z-H)^{-1}_{xy}\right|\leq \frac{1}{\nu d}e^{-\alpha(|\vec x - \vec k|-| \vec k - \vec y|)},
\end{equation}
Setting $y = k, \alpha=\alpha_{\max}$, we recover our estimate on the resolvent.

\section{Proof of \cref{thm:trunc_main}} \label{sec:proof_main}
We can estimate the truncation error by writing out the solutions explicitly
\begin{equation}
\begin{split} \label{eq:total-err}
\|&\zeta(t)\|_\mathcal{H} \\
&=\left\| e^{-iHt}\psi_0 - P_{R}^* e^{-iH_Rt} P_R \mathcal{X}_{B_r}\psi_0 \right\|_{\mathcal H} \\
& =\left\| e^{-iHt}\mathcal{X}_{B_r}\psi_0 + e^{-iHt} \mathcal{X}_{B_r^\complement} \psi_0  - P_{R}^* e^{-iH_Rt} P_R \mathcal{X}_{B_r}\psi_0 \right\|_{\mathcal H} \\
& \leq \left\| \left(e^{-iHt} - P_{R}^* e^{-iH_Rt} P_R \right) P_R\mathcal{X}_{B_r}\psi_0 \right\|_{\mathcal H}  +  \left\| e^{-iHt}\mathcal{X}_{B_r^\complement}\psi_0\right\|_{\mathcal H} \\
 &\leq \underbrace{ \left\|P_R^* \left(P_R e^{-iHt} - e^{-iH_Rt} P_R \right) \mathcal{X}_{B_r}\psi_0 \right\|_{\mathcal H}}_{=:I_1} + 
\underbrace{\left\|\left(I - P_R^*P_R\right) e^{-iHt}\mathcal{X}_{B_r}\psi_0 \right\|_{\mathcal H}}_{=:I_2} + \phi(r).
\end{split}
\end{equation}

Here the last term comes from \cref{as:intial_condition} and the fact that $\exp(-iHt)$ is an isometry. Notice $I_1$ represents the error caused by the truncation of the Hamiltonian, and $I_2$ represents the error of wave functions exiting the truncated domain.
In \cref{sec:estimate1} we have
\begin{equation}
    I_1 \leq \sqrt{\frac{2}{\pi}} C_\gamma \frac{h_0 e^{\delta\alpha_{\max}}}{\nu^2 d^2}\frac{|\Omega_R| |\Omega_r|^{\frac{1}{2}} }{|\Gamma|^{\frac{1}{2}}} C_2^{\frac{1}{2}}(2\alpha_{\max}, R) e^{dt-\alpha_{\max}(R-r)}\left\| \mathcal{X}_{B_r}\psi_0 \right\|_{\mathcal H}, 
\end{equation}
and in \cref{sec:estimate2} we have
\begin{equation}
    I_2 \leq \sqrt{\frac{2}{\pi}} C_\gamma \frac{e^{\delta\alpha_{\max}}}{\nu d} \frac{|\Omega_r|^{\frac{1}{2}}}{|\Gamma|^{\frac{1}{2}}}C_2^{\frac{1}{2}}(2\alpha_{\max}, R) e^{dt-\alpha_{\max}(R-r)} \left\| \mathcal{X}_{B_r}\psi_0  \right\|_{\mathcal H},
\end{equation}
where the constant $C_2$ depends only on $R$ and $\alpha_{\max}$
\begin{equation}\label{eq:C2}
	C_2(\alpha_{\max}, R) := \frac{1+R\alpha_{\max} - \delta\alpha_{\max}}{\alpha_{\max}^2}.
\end{equation} 
The estimates on $\zeta(t)$ follows immediately by summing the terms.

\subsection{Estimation on $I_1$}
\label{sec:estimate1}
We provide two bounds related to the index set $\Omega$ and the truncated index set $\Omega_r$ for TBG, that will be useful in our analysis.

\begin{lemma}
\label{lem:count} 
For $R > r \geq \delta $ and $\alpha_{\max}>0$, we have 
\begin{equation} \label{eq:trunc_decay}
    \sum_{x\in \Omega \setminus \Omega_R} \sum_{y\in \Omega_r } e^{- \alpha_{\max}  |\vec x - \vec y|} \leq  \frac{8\pi e^{2\delta\alpha_{\max}} |\Omega_r|}{|\Gamma|} C_2(R, \alpha_{\max}) e^{-\alpha_{\max}(R-r)},
\end{equation}
where $C_2$ is defined in \cref{eq:C2}. $|\Omega_r|$ is the number of orbitals in $\Omega_r$, which can be estimated by
\begin{equation} \label{eq:count}
\left|\Omega_r\right| = \sum_{x\in \Omega_r} 1 \leq 4 \cdot \left\lceil\frac{4r}{\sqrt{3}a} \right\rceil^2.
\end{equation}
\end{lemma}
\begin{proof} 
To find a sharp bound for the number of lattice points inside a circle of given radius is a well-known problem in number theory. The inequality \cref{eq:count} is estimated by finding the number of unit cells that covers the circle, and multiplying that number by the number of layers and sublattices.

For the \cref{eq:trunc_decay}, we have
\begin{equation}
    \sum_{x\in \Omega \setminus \Omega_R}\sum_{y\in \Omega_r} e^{- \alpha_{\max}  |\vec x - \vec y|} \leq |\Omega_r| \sum_{x\in \Omega \setminus \Omega_R} e^{-\alpha_{\max} (|\vec x| - r)}.
\end{equation}
The summation can be bounded similar to \cref{eq:lattice_sum_bound} by an integral. We here integrate over a larger region, so that all hexagonal unit cells will be contained in the region
\begin{equation}
\begin{split}
 \sum_{x\in \Omega \setminus \Omega_R} e^{-\alpha_{\max} (|\vec x| - r)} 
 &\leq \frac{4 e^{\delta\alpha_{\max}}}{|\Gamma|} \int_{\mathbb{R}^2\setminus B_{R-\delta}} e^{-\alpha_{\max}(|\vec x|-r)} \dee \vec x \\
     &= \frac{8\pi e^{\delta\alpha_{\max}}}{|\Gamma|}\frac{1+R\alpha_{\max} - \delta\alpha_{\max}}{\alpha_{\max}^2}e^{-\alpha_{\max}(R-r)}e^{\alpha_{\max} \delta} \\
     & = \frac{8\pi e^{2\delta\alpha_{\max}} }{|\Gamma|} C_2(\alpha_{\max}, R) e^{-\alpha_{\max}(R-r)}.
\end{split}
\end{equation}

Together with \cref{eq:count} we have the estimate. 
\end{proof}

The eigenvalues of $H_R$ are contained in the spectrum of $H$. 
For the given contour $\gamma$ with $\operatorname{dist}(\gamma, \sigma(H)) > d$, we can write
\begin{equation}
\label{eq:err1_integral}
    \begin{split}
        &P_R^* \left(P_R e^{-iHt} - e^{-iH_Rt} P_R \right) \mathcal{X}_{B_r}\psi_0  \\
        = &   \frac{1}{2\pi i} \int_\gamma e^{-izt} P_R^* \left[P_R (z-H)^{-1} -(z-H_R)^{-1}P_R\right]\mathcal{X}_{B_r}\psi_0  \dee z \\
        =&\frac{1}{2\pi i}  \int_\gamma e^{-izt} P_R^* P_R(z-H)^{-1}(HP_R^*-P_R^*H_R)(z-H_R)^{-1}P_R \mathcal{X}_{B_r}\psi_0  \dee z.
    \end{split}
\end{equation}
Now, note that $HP_R^*-P_R^*H_R = (I - P_R^*P_R)HP_R^*$ from the definition of $H_R$ in \cref{eq:H_trunc}, and the entries of $(I - P_R^*P_R)H$ are explicitly
\begin{equation}
\left[(I - P_R^*P_R)H \right]_{xy} = 
    \begin{cases}
        H_{xy}, \; \text{if } x \in \Omega\setminus\Omega_R, \\
        0, \; \text{otherwise.}
    \end{cases}
\end{equation}

Then we can give a bound on $I_1$ by using a contour integral
\begin{equation}
    \begin{split}
        I_1 
        &= \left\|\frac{1}{2\pi i}  \int_\gamma e^{-izt} P_R^* P_R(z-H)^{-1} (I - P_R^*P_R)HP_R^*(z-H_R)^{-1}P_R \mathcal{X}_{B_r}\psi_0  \dee z \right\|_{\mathcal{H}} \\ 
        &\leq \frac{C_\gamma e^{dt}}{2\pi} \left \| P_R^* P_R(z-H)^{-1} (I - P_R^*P_R)HP_R^*(z-H_R)^{-1}P_R \mathcal{X}_{B_r}\psi_0  \right \|_{\mathcal H}
    \end{split}
\end{equation}
where $C_\gamma$ is the finite length of contour $\gamma$.

Let $u,v,x,y \in \Omega$ be the indices, and $\vec u, \vec v, \vec x, \vec y$ be the respective physical positions, the injection operators $P_R$ and $P_R^*$ allows us to write out the square of norm explicitly through summations over indices. 
\begin{equation}
\label{eq:err1_sum}
\begin{split}
&\left \| P_R^* \left[P_R(z-H)^{-1} (I - P_R^*P_R)HP_R^*(z-H_R)^{-1}P_R\right] \mathcal{X}_{B_r}\psi_0  \right \|_{\mathcal H}^2 \\
    & \leq \sum_{x\in \Omega_R}\sum_{u\in \Omega\setminus\Omega_R}\sum_{v\in \Omega_R} \sum_{y\in \Omega_r} \left|(z-H)^{-1}_{xu}\right|^{2} \left|H_{uv}\right|^{2}\left|(z-H_R)^{-1}_{vy}\right|^2 \left|\psi_0(y) \right|^{2}\\
    & \leq \left\|\mathcal{X}_{B_r}\psi_0 \right\|^2_{\mathcal H} \sum_{x\in \Omega_R}\sum_{u\in \Omega\setminus\Omega_R}\sum_{v\in \Omega_R} \sum_{y\in \Omega_r} \frac{1}{\nu^4d^4} e^{-2\alpha_{\max}|\vec x - \vec u|} h_0^2 e^{-2\alpha_0 |\vec u - \vec v|} e^{-2\alpha_{\max}|\vec v - \vec y|} \\
    & \leq \frac{h_0^2}{\nu^4d^4}\left\|\mathcal{X}_{B_r}\psi_0 \right\|^2_{\mathcal H} \sum_{x\in \Omega_R} \sum_{v\in\Omega_R} \sum_{u\in \Omega\setminus\Omega_R}  \sum_{y\in\Omega_r} e^{-2\alpha_{\max}(|\vec u - \vec v| + |\vec v - \vec y|)}\\
    & \leq \frac{h_0^2}{\nu^4d^4}\frac{8\pi e^{4\delta\alpha_{\max}} |\Omega_R|^2 |\Omega_r|}{|\Gamma| } C_2(2\alpha_{\max}, R) e^{-2\alpha_{\max}(R-r)}\left\|\mathcal{X}_{B_r}\psi_0 \right\|^2_{\mathcal H}.
    \end{split}
\end{equation}
Here we use \cref{thm:combes-thomas} to bound the entries of the resolvent, and \cref{lem:count} to bound the infinite summation. We then take the square root to get the desired result.

\subsection{Estimation on $I_2$}
\label{sec:estimate2}
Similar to the previous estimate, we use contour integral to estimate
\begin{equation}\begin{split}
I_2 = \left\| \left(I - P_R^*P_R\right) e^{-iHt} \mathcal{X}_{B_r}\psi_0  \right\|_{\mathcal H} &=\left\|\frac{1}{2\pi i}\int_\gamma  e^{-izt} \left(I - P_R^*P_R\right)  (z-H)^{-1} \mathcal{X}_{B_r}\psi_0 \right\|_{\mathcal H} \\ 
    &\leq \frac{C_\gamma e^{dt}}{2\pi} \left \| \left(I - P_R^*P_R\right)  (z-H)^{-1} \mathcal{X}_{B_r}\psi_0  \right\|_{\mathcal H}.
\end{split}
\end{equation}
\cref{thm:combes-thomas} and \cref{lem:count} gives
\begin{equation}
\begin{split}
     &\left\| \left(I - P_R^*P_R\right)  (z-H)^{-1} \mathcal{X}_{B_r}\psi_0  \right\|_{\mathcal H}^2  \\
     &= \sum_{x\in\Omega \setminus \Omega_R} \sum_{y\in\Omega_r} |(z-H)^{-1}_{xy}|^2 |\mathcal{X}_{B_r}\psi_0 (y)|^2 \\
    &\leq \sum_{x\in \Omega \setminus \Omega_R}\sum_{y\in \Omega_r} \frac{1}{\nu^2d^2} e^{-2\alpha_{\max}|\vec x - \vec y|} \left\| \mathcal{X}_{B_r}\psi_0  \right\|_{\mathcal H}^2 \\
    &\leq \frac{1}{\nu^2d^2}\frac{8\pi e^{4\delta\alpha_{\max}} |\Omega_r|}{|\Gamma|} C_2(2\alpha_{\max}, R) e^{-2\alpha_{\max} (R-r)}\left\| \mathcal{X}_{B_r}\psi_0  \right\|_{\mathcal H}^2.
\end{split}
\end{equation}
Then the bound on $I_2$ follows immediately after taking the square root.

\end{document}